\documentclass[12pt]{article}
\usepackage{amssymb,amsmath,amsthm,amsfonts,eurosym,geometry,ulem,graphicx,color,setspace,sectsty,comment,footmisc,pdflscape,array,titlesec,hyperref,multicol,multirow,lscape,natbib,tikz,algorithm,algpseudocode,subcaption}
\usepackage[toc,page]{appendix}

\usepackage{indentfirst}
\newtheorem{theorem}{Theorem}

\hypersetup{colorlinks=true,linkcolor=blue,citecolor=blue,urlcolor=blue}

\newcolumntype{L}[1]{>{\raggedright\let\newline\\arraybackslash\hspace{0pt}}m{#1}}
\newcolumntype{C}[1]{>{\centering\let\newline\\arraybackslash\hspace{0pt}}m{#1}}
\newcolumntype{R}[1]{>{\raggedleft\let\newline\\arraybackslash\hspace{0pt}}m{#1}}

\title{Trading Graph Neural Network}
\author{Xian Wu\thanks{University of Wisconsin - Madison, Department of Economics. Email: xwu394@wisc.edu. }}
\date{\today}

\begin{document}

\maketitle
\begin{abstract} This paper proposes a new algorithm -- Trading Graph Neural Network (TGNN) that can structurally estimate the impact of asset features, dealer features and relationship features on asset prices in trading networks.  It combines the strength of the traditional simulated method of moments (SMM) and recent machine learning techniques -- Graph Neural Network (GNN). It outperforms existing reduced-form methods with network centrality measures in prediction accuracy. The method can be used on networks with any structure, allowing for heterogeneity among both traders and assets.
\end{abstract}
\clearpage
\section{Introduction}
\par Many financial markets are organized as trading networks with dealers and customers. Dealers' position in the trading network is shown to have a significant impact on asset prices.\footnote{For example, \cite{di2017value} show that systemically important dealers use their network position to their advantage by charging higher markups to peripheral dealers and clients than to other core dealers. \cite{hollifield2017bid} find that core dealers receive relatively lower and less dispersed spreads than peripheral dealers in securitization markets. \cite{li2019dealer} find that core dealers charge customers up to twice the round-trip markups compared to peripheral dealers in municipal bond markets.} However, it remains challenging to account for the structure of trading networks during the estimation of dealer and asset features' impact on asset prices. Structural approaches usually rely on specific network structures to reduce complexity in estimation 
 \citep[e.g.][]{pinter2022comparing,eisfeldt2023otc,cohen2024inventory}, which limits the accuracy and generalizability of the estimation method. Reduced-form approach uses centrality measures to capture dealers' position in the network\citep[e.g.][]{di2017value,hollifield2017bid,li2019dealer}, but recent papers point out linear regressions with centrality measures can lead to biased estimation when the network is sparse \citep{cai2022linear}. This paper addresses these limitations with a novel structural estimation method -- Trading Graph Neural Network (TGNN). It disentangles the impact of asset, trader features, and traders' relationships on asset prices in arbitrary networks. 
 \par The Trading Graph Neural Network (TGNN) builds on a parsimonious trading network model with dealers as nodes and dealers' connections as edges. Each dealer's value is the difference between their maximum resale price and its holding cost. The potential price each dealer seller can get from another dealer buyer is the average of their values weighted by their bargaining powers based on the Nash-bargaining solution. The dealers have the outside option to sell to their customers if they get lower prices in the interdealer network. We show that the model has unique bounded solution that can be found by applying a contraction mapping iteratively given customer values, dealer bargaining powers, and holding costs, which can be parameterized and estimated with data. 
 \par TGNN estimates the model by identifying the impact of asset features, dealers' features and their relationship on the customers values, dealer bargaining powers and holding costs. We start with random initial parameters on asset features, dealers' features, and their relationship to simulate traders' values and find the predicted prices with the parameterized model. We compute the mean squared error loss by comparing the predicted prices and observed prices. We then update the parameters with backpropagation, a technique commonly used in machine learning, repeat the above process until the loss is minimized. TGNN further provides confidence intervals of the estimated parameters with the bootstrap method.
 \par TGNN combines the strength of the traditional simulated method of moments (SMM) and recent machine learning techniques -- Graph Neural Network (GNN). TGNN resembles the classic SMM in the use of structural estimation. Compared with SMM, it improves the estimation efficiency by adopting machine learning techniques, and reduces potential bias in moment selections by directly comparing observable prices and estimated prices. TGNN adapted GNN with economic micro-foundations of its key message passing framework. In TGNN, the message passing process coincides with the contraction mapping to find the fixed point of dealers' values, and can be interpreted as rounds of request-for-quote and updates of dealers' values until convergence. 
 \par We provide several test cases to examine the performance of TGNN. The estimates of the parameters on the asset features, dealer features and relationship features with TGNN are close to the true parameters. It can recover the bargaining power, dealer holding cost, dealer values and potential transaction prices with high accuracy. It outperforms the traditional reduced-form approach with centrality measures in explaining the observed prices in dense, sparse random networks and core-periphery networks. 
 \par TGNN enables structural estimation from transaction data, linking OTC market theory with empirical analysis. It captures heterogeneity in trading relationships and produces economically interpretable parameters for counterfactual analysis. It can be applied in analysis of traditional OTC markets like bonds or interbank lending, or decentralized digital markets such as crypto trading and peer-to-peer lending to identify key participants, estimate bargaining power and holding costs, assess price formation, market power, and systemic risk. 

\paragraph{Literature} This paper is related to the literature studying the impact of asset, trader features and their relationship on asset prices in a trading network. The majority of empirical papers in this strand of literature take a reduced-form approach, using centrality measures to capture the impact of network structure \citep{di2017value,hollifield2017bid,li2019dealer}. Theoretical models with general network structure are rarely used for structural estimation, due to additional assumptions on network structures that may not hold in the data or difficulty in the estimation of information structure \citep{malamud2017decentralized,babus2018trading,rostek2025public}. Among the papers structurally estimating the impact of dealer features on asset prices, most papers adopt the search and matching framework \citep[e.g.][]{pinter2022comparing,cohen2024inventory}, i.e. assuming the trading network is pairwise bilateral. \cite{eisfeldt2023otc} evaluate how changes in network structure will affect the price of credit default swaps (CDS) in core-periphery networks. 
This paper contributes to the literature with a new approach to structurally estimate the impact of asset, trader features and their relationship in networks of arbitrary structures.
\par This paper is related to the literature of the Simulated Method of Moments (SMM). SMM estimates economic model parameters by matching simulated model moments to data moment \citep{mcfadden1989method,pakes1989simulation}. It is especially useful for models where analytical solutions are difficult. However, SMM can be computationally intensive with complex network models. This paper proposes an algorithm with machine learning techniques to improve the computation efficiency. Also, our method directly matches the predicted prices with the observed prices. This avoids choosing the moments which can affect the accuracy of the estimates.
\par  This paper is also related to the growing literature using the Graph Neural Network (GNN) to study social networks and financial networks. \cite{leung2022graph} apply GNN to adjust for high-dimensional network confounding in causal inference, say peer effects in selection into treatment in a network. \cite{wang2024graph} use GNN to estimate network heterogeneity, examine the average treatment effects and outcomes of counterfactual policies, and to select the initial recipients of program information in social networks. Both \cite{leung2022graph} and \cite{wang2024graph} focus on GNN's application in addressing confounding factors in causal inference in networks, but maintain the key element used in GNN -- the message-passing framework in its general form without providing economic foundations. \cite{brogaard2024attention} apply GNN to an inter-firm network to predict credit default swap (CDS) spreads. It finds that GNN with firm dynamics through networks can improve the prediction accuracy relative to traditional machine learning methods without edge features. While \cite{brogaard2024attention} use the attention-based GNN to infer the firms or inter-firm linkages with the highest impacts on CDS spreads, the message passing framework remains a black-box, creating difficulties in quantifying and interpreting the impacts of the firms' features. Compared with the existing papers, this paper micro-founds the message-passing framework in GNN with an economic model, and provides direct estimates of the impact of node and edge features with confidence intervals. 

\section{Model}\label{sec:model}
\par Let $G$ be a network with $n$ dealers, where each edge denotes a trading relationship. Denote the set of dealers that a dealer $i$ can sell to as $\mathcal{N}(i)$.
The dealer $i$'s value for the asset $k$ at time $t$ satisfies
\begin{equation}\label{eq:value}
    v_{ikt} =  - c_{ikt} + \max\{\max_{j\in \mathcal{N}(i)} \{p_{ijkt}\}_j ,u_{ikt}\}
\end{equation}
where $c_{ikt}>0$ is the cost for dealer $ i$ to hold the asset $k$ at time $t$. $u_{ikt}$ is the price for dealer $i$ to sell it to its customers, and we will refer to it as customer value.
\par $p_{ijkt}$ is the transaction price of the asset $k$ if the dealer $i$ and $j$ trade at time $t$. It follows the Nash-bargaining solution that 
\begin{equation}\label{eq:price}
    p_{ijkt} = \pi_{ijt}v_{ikt}+(1-\pi_{ijt})v_{jkt}.
\end{equation}
where $\pi_{ijt}\in(0,1)$ is the bargaining power of dealer $j$ to buy from dealer $i$ at time $t$. Intuitively, when the dealer $j$ has a larger bargaining power, the dealer $j$ has a larger share of the surplus of the trade $v_{jkt}-v_{ikt}$.
    

\begin{theorem}[Equilibrium Existence and Uniqueness]\label{thm:eqlm} There exists unique fixed point of equilibrium values $\{v_{ikt}^*\}_{i,k,t}$ given $\{c_{ikt},u_{ikt}\}_{i,k,t}$ and $\{\pi_{ijkt}\}_{i,j,k,t}$. And for each asset $k$ and date $t$, we can find the equilibrium value $\boldsymbol{v}_{kt}^*=(v_{1kt}^*, \dots, v_{nkt}^*) \in \mathbb{R}^n$ with the following iterative algorithm. We can start with an arbitrary element $\boldsymbol{v}_{kt}^{(0)}\in\mathbb{R}^n$, and define a sequence $\{\boldsymbol{v}_{kt}^{(n)}\}_{n\in \mathbb{N}}$ by $\boldsymbol{v}_{kt}^{(n)}=T(\boldsymbol{v}_{kt}^{(n-1)})$ for $n\geq 1$ where $T: \mathbb{R}^n \to \mathbb{R}^n$ is a contraction mapping whose $i^{th}$ component is $ T_i(\boldsymbol{v}) = -c_i + \max\left\{ \max_{j \in \mathcal{N}(i)} \left\{ \pi_{ijkt} v_{ikt} + (1 - \pi_{ijkt}) v_{jkt} \right\}, u_{ikt} \right\}$. Then $\lim_{n\rightarrow \infty} \boldsymbol{v}_{kt}^{(n)} = \boldsymbol{v}_{kt}^*$. 
\par The equilibrium values are bounded by $v_{min}^*\equiv\min_i\{v_{ikt}\}_{i,k,t}\geq \min_i\{u_{ikt}-c_{ikt}\}$ and $v_{max}^*\equiv\max_i\{v_{ikt}\}_{i,k,t}= \max_i\{u_{ikt}-c_{ikt}\}$.
\par In equilibrium, seller $i$ sells the asset $k$ at price $ p_{ikt}^*=\max\left\{ \max_{j \in \mathcal{N}(i)} \left\{ \pi_{ijkt} v_{ikt}^* + (1 - \pi_{ijkt}) v_{jkt}^* \right\}, u_{ikt} \right\}$ at time $t$, to counterparty $j=\arg\max_{j \in \mathcal{N}(i)} \left\{ \pi_{ijkt} v_{ikt}^* + (1 - \pi_{ijkt}) v_{jkt}^* \right\}$ if $p_{ikt}^*>u_{ikt}$, and to customers otherwise.
\end{theorem}

\paragraph{Remark} Theorem \ref{thm:eqlm} establishes the existence, uniqueness and finiteness of the equilibrium. It also provides a iterative contraction mapping to find the fixed point of dealer values. This allows us to simulate the equilibrium given the customer values, dealer holding costs and bargaining power. We can parameterize the customer values, dealer holding costs and bargaining power as functions of asset features, dealer features and relationship features. With the parameterization, we can structurally estimate the impact of these features on the customer values, dealer holding costs, bargaining power, and therefore asset prices. 

\paragraph{Parameterization} We assume that the holding cost satisfies
\begin{equation}\label{eq:cost}
    c_{ikt} = f(\boldsymbol{X}_{kt}'\boldsymbol{\beta}_x + \boldsymbol{Y}_{it}' \boldsymbol{\beta}_y + \varepsilon_{ikt})
\end{equation}
where $f(\cdot):\mathbb{R}\mapsto \mathbb{R}^+$, e.g. the exponential function. $\boldsymbol{X}_{kt}$ are asset features at time $t$, e.g. term to maturity, bond type or bond fixed effects. $\boldsymbol{Y}_{it}$ are dealer $i$'s features at time $t$, e.g. the inventory of investment grade bonds and high yield bonds at time $t$, dealer fixed effects and date fixed effects. 
\par The value to sell the asset to dealer $i$'s customer is 
\begin{equation} \label{eq:customer_val}
   u_{ikt}= g(\boldsymbol{X}_{kt}'\boldsymbol{\gamma}_x + \boldsymbol{Y}_{it}' \boldsymbol{\gamma}_y +\boldsymbol{Z}_{it}' \boldsymbol{\gamma}_z +\zeta_{ikt})
\end{equation}
\par where $g(\cdot):\mathbb{R}\mapsto \mathbb{R}^+$, e.g. the exponential function. $\boldsymbol{Z}_{it}$ are dealer-customer features, e.g. the number of past trades and past buy and sell volume with customers of dealer $i$ in a given period.   
\par The bargaining power between the two dealers satisfies 
\begin{equation}\label{eq:bargaining}
    \pi_{ijt}= \sigma(\boldsymbol{E}_{ijt}'\boldsymbol{\eta} +\epsilon_{ijt})
\end{equation}
where $\sigma(\cdot):\mathbb{R}\mapsto [0,1] $, e.g. the logistic function. $\boldsymbol{E}_{ijt}$ are relationship features, e.g. the number of past trades, and past trading volume between dealer $i$ and dealer $j$, edge and date fixed effects. 

\section{Estimation}

\paragraph{Estimate Customer Values}
\par Note that when an interdealer trade happens between dealer seller $i$ and buyer $j$, we typically do not observe a simultaneous sale of the same asset between dealer $i$ and its customers. We can first estimate the following equation with the transactions of dealer-to-customer sales, 
\begin{equation*}
p_{ikt}^c= g(\boldsymbol{X}_{kt}'\boldsymbol{\gamma}_x + \boldsymbol{Y}_{it}' \boldsymbol{\gamma}_y +\boldsymbol{Z}_{it}' \boldsymbol{\gamma}_z +\zeta_{ikt})
\end{equation*}
where $p_{ikt}^c$ is the price of asset $k$ from trader $i$ to its customer at time $t$.  Given the estimated parameters $(\hat{\boldsymbol{\gamma}}_x,\hat{\boldsymbol{\gamma}}_y,\hat{\boldsymbol{\gamma}}_z)$, we can estimate the latent customer value $\hat{u}_{ikt}=g(\boldsymbol{X}_{kt}'\hat{\boldsymbol{\gamma}}_x + \boldsymbol{Y}_{it}' \hat{\boldsymbol{\gamma}}_y +\hat{\boldsymbol{Z}}_{it}' \boldsymbol{\gamma}_z)$ when such trades are not observable.

\paragraph{Estimation with the Trading Network} Given the estimated customer value $\hat{u}_{ikt}$, we can estimate the other parameters with Algorithm \ref{alg:train}, Trading Graph Neural Network (TGNN).  The Trading Graph Neural Network (TGNN) estimates price formation and trading relationships in a financial market represented as a graph, with the underlying data generation process modeled as Section \ref{sec:model}. Given a set of dealers as nodes and trading relationships as edges, the model assigns initial values of dealer values and costs based on asset and trader features. It iteratively updates these values by passing the potential transaction prices information through the network, where transaction prices are computed as a weighted combination of dealer values based on bargaining power. Each dealer updates its value by selecting the best available price among its neighbors. The model is trained by minimizing the mean squared error between predicted and observed best prices, with optional regularization to avoid overfitting. Parameter updates are performed using gradient descent. This framework estimates the impact of asset features, trader features, and the trading relationship on asset prices, taking into account the trading network structure.
\par The algorithm draws inspiration from the Simulated Method of Moments (SMM) as well as the Graph Neural Network (GNN).
\par Simulated Method of Moments (SMM) is widely used for structural estimation of economic models \citep{mcfadden1989method,pakes1989simulation}. TGNN incorporates elements similar to SMM by explicitly modeling the data generation process as outlined in Section \ref{sec:model}, but differs in how it estimates parameters from observed data. While SMM involves simulating data from the model for different parameter values and matching statistical moments with observed data, TGNN directly minimizes the difference between predicted and observed prices. SMM often uses simulation-based search algorithms, sometimes with numerical approximations of gradients. It can be computationally intensive due to repeated simulations, especially with complex models. TGNN uses neural network optimization techniques like Adam, which can be more computationally efficient for graph-based complex models due to gradient-based optimization.
\par In recent years, the Graph Neural Network (GNN) has emerged as a powerful class of deep learning models for processing graph-structured data. Standard GNNs operate through an iterative message passing framework where node representations are updated by aggregating information from their neighbors. Formally, at iteration $t$, a node $v$ updates its representation $h_v^{(t)}$ according to the following equation,
\begin{equation}\label{eq:GNN}
    h_v^{(t)} = \text{UPDATE}\left(h_v^{(t-1)}, \text{AGGREGATE}\left(\{h_u^{(t-1)}: u \in \mathcal{N}(v)\}\right)\right)
\end{equation}
where $\mathcal{N}(v)$ denotes the neighbors of node $v$. The AGGREGATE function combines the features of a node’s neighbors, often using operations like sum, mean, or max after linear or nonlinear transformation of the neighboring nodes' features. The UPDATE function then combines this aggregated message with the node’s previous features to produce its new representation, typically through a neural network with non-linear activation.
\par TGNN shares this iterative information propagation structure in equation (\ref{eq:GNN}) but differs in several important aspects. Both models leverage the network topology to capture dependencies between entities and employ multiple iterations of message passing to refine representations. However, the TGNN replaces generic neural network components with economically-motivated functions, implementing a domain-specific message passing scheme derived from bargaining theory. While standard GNNs use learnable parameters that lack direct interpretation, TGNN explicitly models dealer costs, customer values, and bargaining powers. Most importantly, the message passing framework in equation (\ref{eq:GNN}) is micro-founded by Theorem \ref{thm:eqlm}. Node updates in our model follow economic maximization principles rather than generic transformations, as dealer values are updated by taking the maximum between customer value and the best available price from other dealers (see equation (\ref{eq:value})). The message passing framework resembles the contraction mapping $T(\cdot)$ in Theorem \ref{thm:eqlm}. It can be interpreted as iterated belief updates with rounds of request-for-quote from each seller to its potential buyer until the dealers' values converge. This economic structure improves interpretability while retaining the powerful representational capabilities of graph-based learning. Finally, general GNNs are usually used for prediction tasks, and can often sacrifice the interpretability of the model in exchange for the accuracy of prediction. However, TGNN focuses on the estimation of the parameters, the economic interpretation, and its potential applications in estimating counterfactuals that are interesting to economists, such as removing dealers from the market,  increasing dealer inventory cost, or restructuring the trading network.
\paragraph{Bootstrap Confidence Intervals} As we focus more on the structural estimation, we introduce the bootstrap method to calculate confidence intervals as illustrated in Algorithm \ref{alg:CI}, which are not typically used in general GNN.  
The bootstrap method, introduced by \cite{efron1992bootstrap}, is a computational resampling technique used to estimate the sampling distribution and quantify uncertainty of parameter estimates without relying on strong distributional assumptions. In our trading network application, we implement a modified bootstrap approach that preserves the network structure while resampling the observed prices. Specifically, for each bootstrap iteration $b = 1,2,\ldots,B$, we generate a bootstrap sample by resampling with replacement from the original set of observed prices. We then train our graph neural network model on this modified data to obtain bootstrap parameter estimates $\hat{\beta}_x^{(b)}$, $\hat{\beta}_y^{(b)}$, and $\hat{\eta}^{(b)}$. From the resulting empirical distribution of parameter estimates $\{\hat{\theta}^{(1)}, \hat{\theta}^{(2)}, \ldots, \hat{\theta}^{(B)}\}$, we calculate $(1-\alpha)$ confidence intervals using the percentile method, $\text{CI}_{1-\alpha}(\hat{\theta}) = [\hat{\theta}_{(\alpha/2)}, \hat{\theta}_{(1-\alpha/2)}]$, providing a measure of the statistical uncertainty in our parameter estimates that arises from the specific set of observed prices.
\par Using the bootstrap method to calculate the confidence interval has two benefits. First, it does not rely on the asymptotic normality of the estimator. Second, it is computationally stable than estimation with the asymptotic covariance matrix when the data generation function is not continuous \citep{mcfadden1989method,pakes1989simulation}. It is particularly useful in this model, as each trader's optimization involves discrete changes in values across its local network, and the data generation function can be non-differentiable. 

\begin{algorithm}
\caption{Trading Graph Neural Network (TGNN)}\label{alg:train}
\begin{algorithmic}[1]
\State \textbf{Input:} Graph $G = (\mathcal{V}, \mathcal{E})$ with asset features $\boldsymbol{X}$, dealer features $\boldsymbol{Y}$, relationship features $\boldsymbol{E}$, customer values $u$, observed best prices $p^{obs}$
\State \textbf{Parameters:} $\boldsymbol{\beta}_x$, $\boldsymbol{\beta}_y$, $\boldsymbol{\eta}$
\State \textbf{Hyperparameters:} Number of message passing iterations $L$, learning rate $\alpha$, regularization parameter $\lambda$
\State Initialize the parameters with small random values
\State \textbf{Forward Pass:}
\State Compute costs: $c_i = f(\boldsymbol{X}_i'\boldsymbol{\beta}_x +  \boldsymbol{Y}_i'\boldsymbol{\beta}_y)$ for all nodes $i \in \mathcal{V}$
\State Compute bargaining powers: $\pi_{ij} = \sigma( \boldsymbol{E}_{ij}'\boldsymbol{\eta})$ for all edges $(i,j) \in \mathcal{E}$
\State Initialize dealer values: $v_i^{(0)} = -c_i + u_i$ for all nodes $i \in \mathcal{E}$
\For{$l = 1$ to $L$}
    \State Compute transaction prices: $p_{ij}^{(l)} = \pi_{ij} \cdot v_i^{(l-1)} + (1-\pi_{ij}) \cdot v_j^{(l-1)}$ for all edges $(i,j) \in \mathcal{E}$
    \For{each node $i \in \mathcal{V}$}
        \State Find best price: $p_{i,best}^{(l)} = \max_{j \in \mathcal{N}_i} p_{ij}^{(l)}$ where $\mathcal{N}_i$ is the set of neighbors of $i$
        \State Update dealer value: $v_i^{(l)} = -c_i + \max\{u_i, p_{i,best}^{(l)}\}$
    \EndFor
\EndFor
\State Identify predicted best prices $p_{i,best}^{pred}$ for each dealer $i$
\State \textbf{Loss Computation:}
\State Compute Mean Squared Error (MSE) loss: $\mathcal{L}_{MSE} = \frac{1}{|\mathcal{V}|} \sum_{i \in V} (p_{i,best}^{pred} - p_{i,best}^{obs})^2$
\State Compute regularization (optional): $\mathcal{L}_{reg} = \lambda(\|\boldsymbol{\beta}_x\|_2^2 + \|\boldsymbol{\beta}_y\|_2^2 + \|\boldsymbol{\eta}\|_2^2)$
\State Total loss (optional): $\mathcal{L} = \mathcal{L}_{MSE} + \mathcal{L}_{reg}$

\State \textbf{Backward Pass:}
\State Compute gradients: $\nabla_{\beta_x}\mathcal{L}$, $\nabla_{\beta_y}\mathcal{L}$, $\nabla_{\eta}\mathcal{L}$
\State Update parameters:
\State $\boldsymbol{\beta}_x \leftarrow \boldsymbol{\beta}_x - \alpha \cdot \nabla_{\beta_x}\mathcal{L}$
\State $\boldsymbol{\beta}_y \leftarrow \boldsymbol{\beta}_y - \alpha \cdot \nabla_{\beta_y}\mathcal{L}$
\State $\boldsymbol{\eta}\leftarrow \boldsymbol{\eta} - \alpha \cdot \nabla_{\eta}\mathcal{L}$
\State \textbf{Return:} Updated parameters $\boldsymbol{\beta}_x$, $\boldsymbol{\beta}_y$, $\boldsymbol{\eta}$
\end{algorithmic}
\end{algorithm}

\begin{algorithm}
\caption{Bootstrap Method to Calculate Confidence Intervals}\label{alg:CI}
\begin{algorithmic}[1]
\State \textbf{Input:} Trained model $\hat{\theta}$, data $\mathcal{D}$, number of bootstrap samples $B$, significance level $\alpha$
\State \textbf{Output:} Confidence intervals for $\beta_x$, $\beta_y$, and $\eta$
\State Extract original parameter estimates $\hat{\beta}_x$, $\hat{\beta}_y$, $\hat{\eta}$ from $\hat{\theta}(\mathcal{D})$
\State Create bootstrap samples from data $\mathcal{D}$ by resampling observable prices while maintaining the network structure
\State Train the model on each bootstrap sample with the same architecture as $\hat{\theta}$
\State Collect estimated parameters $\mathcal{B}_{\beta_x}=\{\hat{\beta}_x^{(b)}\}_{b=1}^B$, $\mathcal{B}_{\beta_y}=\{\hat{\beta}_y^{(b)}\}_{b=1}^B$, and $\mathcal{B}_{\eta}=\{\hat{\eta}^{(b)}\}_{b=1}^B$

\For{each parameter $\phi \in \{\beta_x, \beta_y, \eta\}$}
    \State Compute bootstrap mean: $\bar{\phi} = \frac{1}{B}\sum_{b=1}^{B} \phi^{(b)}$
    \State Compute bootstrap standard error: $\text{SE}(\phi) = \sqrt{\frac{1}{B-1}\sum_{b=1}^{B}(\phi^{(b)} - \bar{\phi})^2}$
    \State Compute $(1-\alpha)$ confidence interval: 
    \State \hspace{1em} $\text{CI}_{\text{lower}}(\phi) = \text{Percentile}(\mathcal{B}_{\phi}, \alpha/2)$
    \State \hspace{1em} $\text{CI}_{\text{upper}}(\phi) = \text{Percentile}(\mathcal{B}_{\phi}, 1-\alpha/2)$
\EndFor
\State \textbf{Return:} Confidence intervals for each parameter dimension
\end{algorithmic}
\end{algorithm}

\section{Test Cases}
In this section, we provide test cases to evaluate the performance of our model in explaining the observed prices and recovering the true parameters. 

\subsection{Network Structure and Synthetic Data Generation}

Our trading network model represents an over-the-counter market with dealers trading assets over multiple days. In this test case, the network is constructed as follows:

\paragraph{Network Dimensions:} We create a synthetic trading environment with 10 dealers ($N_i = 10$) and 2 different assets ($N_k = 2$) observed over 6 trading days ($N_t = 5$).  For simplicity, we assume that the asset features $\boldsymbol{X}$, dealer features $\boldsymbol{Y}$ and trading relationship $\boldsymbol{E}$ are 1-dimensional.  We allow the network to change over time to highlight the flexibility of the algorithm, but this step is not necessary in applications. 

\paragraph{Node Generation:} The sample contains a total of $N_i \times N_k \times N_t = 10 \times 2 \times 5 = 100$ nodes, where each node represents a dealer-asset-day combination. Each node $i$ is assigned index $ikt$ identifying which dealer, asset, and day it represents. Node features are generated as follows. Each asset $k$ at time $t$ has feature $X_{kt} \sim \mathcal{N}(0, 1)$ are drawn from standard normal distributions. Each dealer $i$ at time $t$ has feature $Y_{it} \sim \mathcal{N}(0, 1)$ are drawn from a standard normal distribution. For simplicity, we assume $X_{kt}$ are independent across assets and time, $Y_{it}$ are independent across dealer and time, and $X_{kt}$ and $Y_{it}$ are independent of each other. Customer values $u$: We skip the estimation of customer value and assume $u_{ikt}$ are drawn from a log-normal distribution $exp(Z_{ikt}+5)$ where $Z_{ikt}\sim \mathcal{N}(0,0.01)$. 
\paragraph{Edge Structure:} Edges represent potential trading relationships between dealers for the same asset on the same day. For each day $t$ and asset $k$, we consider all ordered pairs of dealers $(i,j)$ where $i \neq j$. We generate Erdos–Rényi (ER) random graphs where each potential edge is included with probability 70\%, resulting in an undirected sparse network where each dealer can potentially trade with multiple other dealers. Each edge $(i,j)$ is assigned relationship features drawn from independent standard normal distribution $E_{ij} \sim \mathcal{N}(0, 1)$. 

\paragraph{Latent Variables:} Based on the network structure and features, we generate key economic variables according to our model:
\begin{itemize}
    \item Dealer costs $c$: Following equation (\ref{eq:cost}), let $c_{ikt} = \exp(X_{kt}'\beta_x + Y_{it}'\beta_y + \epsilon_{ikt})$ where $\epsilon_{ikt} \sim \mathcal{N}(0, 0.01)$ represents a small noise term and the true parameters are $\beta_x = \beta_y = 1$. 
    \item Bargaining powers $\pi$: Following equation (\ref{eq:bargaining}), $\pi_{ij} = \sigma(E_{ij}'\eta + \nu_{ij})$, where $\sigma$ is a logistic function, $\nu_{ij} \sim \mathcal{N}(0, 0.01)$ is a small noise term, and the true parameter is $\eta = 1$.
\end{itemize}

\paragraph{Dealer Value and Price Determination:} In step 1, For each dealer-asset-day combination, we initiate each dealer's value for the asset as the sum of the cost and the customer value,
$$
v_{ikt} = -c_{ikt} + u_{ikt}.
$$
\par In step 2, we calculate the potential transaction price between connected dealers following equation (\ref{eq:price}). In step 3, given the transaction price $\{p_{ijkt}\}_{ijkt}$, we can update the dealer value $v_{ikt}$ according to equation (\ref{eq:value}). In step 4, we repeat step 2-3 for $L$ times until the changes in $\{v_{ikt}\}_{ikt}$ are sufficiently small. In this exercise, we take the number of iterations $L$ to be $10$. Appendix Figure \ref{fig:value_changes_dense} shows the evolution of the dealer values. Finally, for each dealer, asset and day, we identify and record the best (highest) price offered by any potential counterparty. This observed best price is what the model uses for parameter estimation, reflecting the missing data of unrealized trades in real trading networks.
\medskip 
\par Figure \ref{fig:test_network} shows the structure of interdealer networks generated for each asset on each day. The blue lines indicate the edges of the trading networks. The red lines indicate the observable prices and trading direction. The buyer and seller can change across assets from day to day given different asset and trader features. It's less likely to have interdealer trades when the customer values are high enough and the cost to hold the asset is high for dealers. 

\begin{figure}[htbp]
    \includegraphics[width=\linewidth]{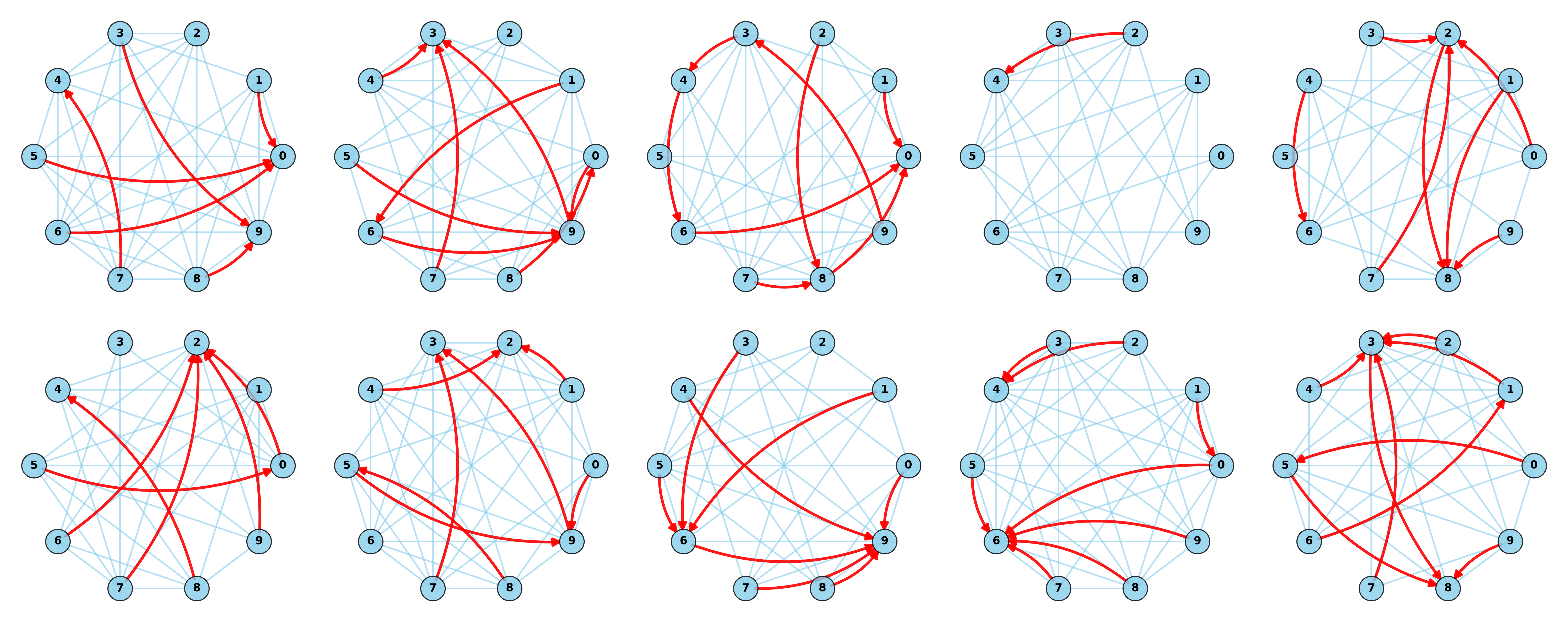}
    \caption{The Structure of Dense Random Networks}
    \label{fig:test_network}
     \footnotesize{\textit{Note:} This figure shows the structure of interdealer networks generated for 10 dealers, 2 assets over 5 days with 622 edges. The blue lines indicate the directed edges, where the edge is directed from seller to buyer. The red arrows denote observable trades, with the arrowheads indicating the direction toward the buyer. \raggedright}
\end{figure}   

\par Table \ref{tab:summary_stats_dense} shows the summary of statistics. Panel A presents the statistics of observable variables that will be used to identify the parameters with TGNN. Panel A presents the statistics of unobservable latent variables in the data generation process. 

\begin{table}[htbp]
\caption{Summary Statistics in Dense Random Networks}
\begin{center}
    \begin{tabular}{lccccc}
\hline
Variable & N & Min & Max & Mean & Std \\ \hline
\textit{Panel A. Observable Variables}\\
\hline
Asset Feature X & 100 & -1.1229 & 2.2082 & 0.1920 & 0.8205 \\ 
Dealer Feature Y & 100 & -2.2064 & 2.8140 & 0.0969 & 0.9887 \\ 
Relationship Feature E & 622 & -2.4396 & 2.7421 & -0.0097 & 0.9736 \\ 
Customer Values & 100 & 148.5882 & 163.8168 & 156.5560 & 4.3693 \\ 
Observed Prices & 68 & 151.1072 & 162.7155 & 160.1135 & 2.4620 \\ 
\hline
\textit{Panel B. Unobservable Latent Variables}\\
\hline
Dealer Values & 100 & 71.2753 & 163.0024 & 156.2776 & 11.2075 \\ 
Bargaining Powers & 622 & 0.0735 & 0.9398 & 0.4967 & 0.2057 \\ 
Potential Transaction Prices & 622 & 79.4330 & 162.9168 & 156.4676 & 9.4779 \\ 
Costs & 100 & 0.0758 & 82.7652 & 3.8293 & 9.7679 \\ 
\hline
\end{tabular}
\end{center}
\label{tab:summary_stats_dense}
\footnotesize{\textit{Note:}  This table shows the generated data of observable and latent variables in the test case of dense random networks. } 
\end{table}


\subsection{Evaluating the Performance}
\par We apply TGNN to estimate the parameters on asset features, dealer features, and relationship features with the dense random networks generated above.  We initialize the parameters with random values from a uniform distribution between $[-0.1,0.1]$. We calculated the Mean Squared Error (MSE) loss between the predicted prices and observed prices. The number of message passing iterations $L$ is 10, and the learning rate $\alpha$ is 0.01. We trained the model for 300 epochs.\footnote{ Appendix Figure \ref{fig:test_train_loss} shows the training loss over time. The loss converges when the training epoch is larger than 200.} 
\par We use bootstrap method to compute the confidence interval. Our approach maintains the fixed network structure of dealer-asset relationships while resampling from the set of observed best prices. For each bootstrap iteration, we randomly sample a subset of prices with replacement, retrain the model on this modified dataset, and extract the resulting parameter estimates. This process is repeated 100 times to generate empirical distributions for each parameter vector. 
\par Figure \ref{fig:est_dense} shows the estimated parameters with confidence intervals. First, the estimates are close to the true values which are within 95\% confidence intervals (CIs). Second, the distribution of the estimates can be skewed. Third, the bootstrap results show that relationship parameters exhibit the widest confidence bands, indicating greater sensitivity to which specific prices are observed. These patterns align with our theoretical understanding that asset-specific factors are identified through multiple observations of the same asset across different dealers, while relationship-specific factors depend more heavily on observing multiple observations of the same buyer-seller pair across different assets and days.
\par Figure \ref{fig:est_latent_dense} shows the predicted unobservable latent variables against their true values. We can see that the TGNN recovers unobservable bargaining power, holding costs, dealer values, and potential transaction prices between two dealers with high accuracy. 

\begin{figure}
     \centering
     \begin{subfigure}[b]{0.6\textwidth}
         \centering
         \includegraphics[width=\textwidth]{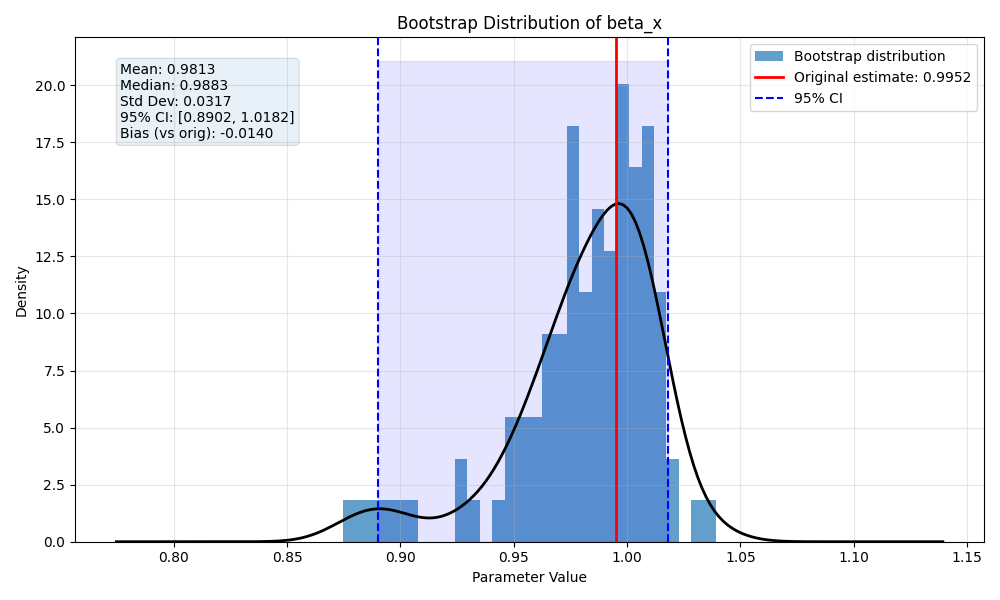}
         \caption{$\hat{\beta}_x$}
         \label{fig:beta_x_dense}
     \end{subfigure}
     \hfill
     \begin{subfigure}[b]{0.6\textwidth}
         \centering
         \includegraphics[width=\textwidth]{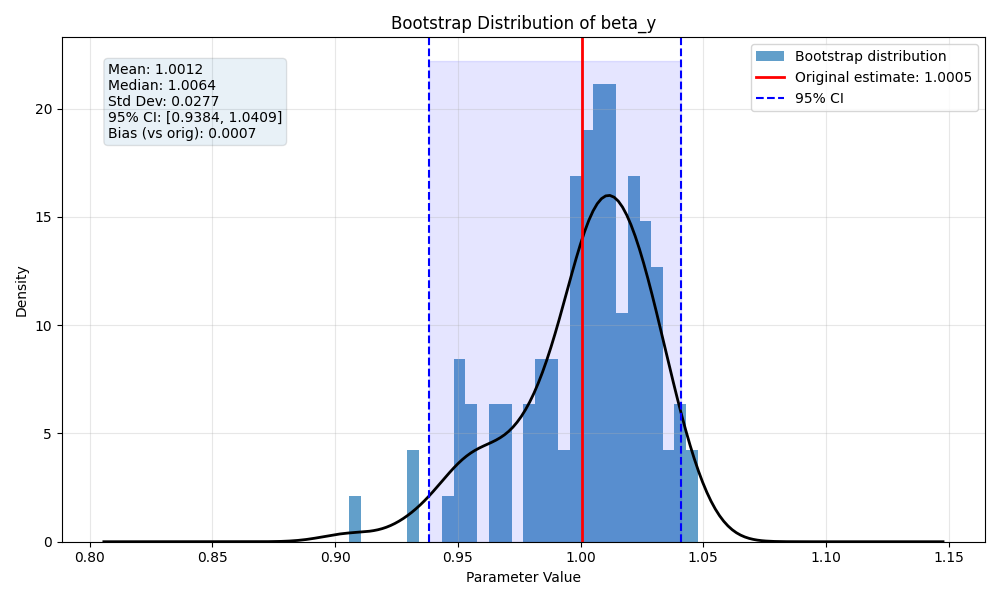}
         \caption{$\hat{\beta}_y$}
         \label{fig:beta_y_dense}
     \end{subfigure}
     \hfill
     \begin{subfigure}[b]{0.6\textwidth}
         \centering
         \includegraphics[width=\textwidth]{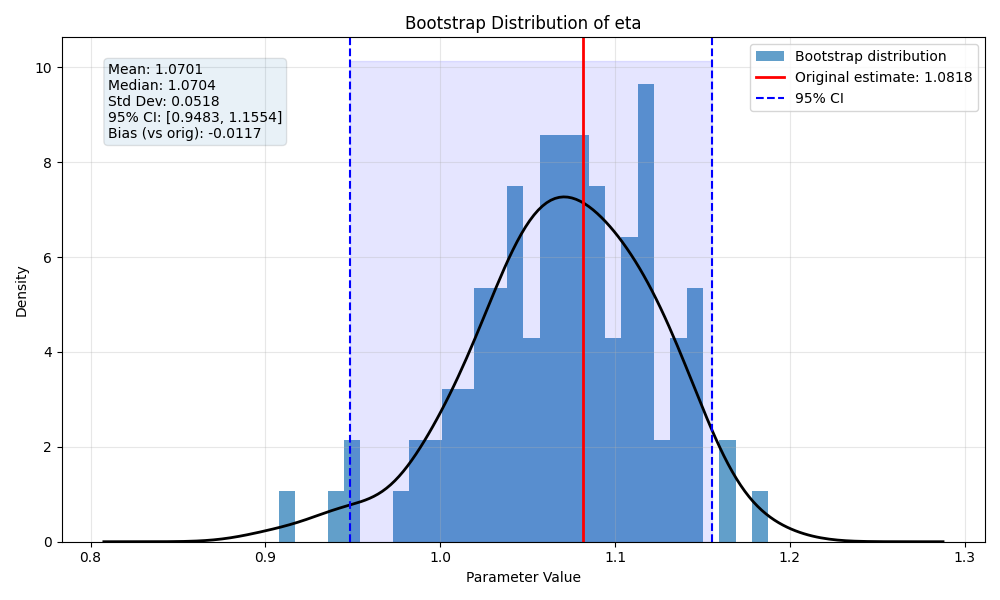}
         \caption{$\hat{\eta}$}
         \label{fig:eta_dense}
     \end{subfigure}
        \caption{Estimated Parameters in Dense Random Networks}
        \label{fig:est_dense}
\end{figure}

\begin{figure}
     \centering
     \begin{subfigure}[b]{0.5\textwidth}
         \centering
         \includegraphics[width=\textwidth]{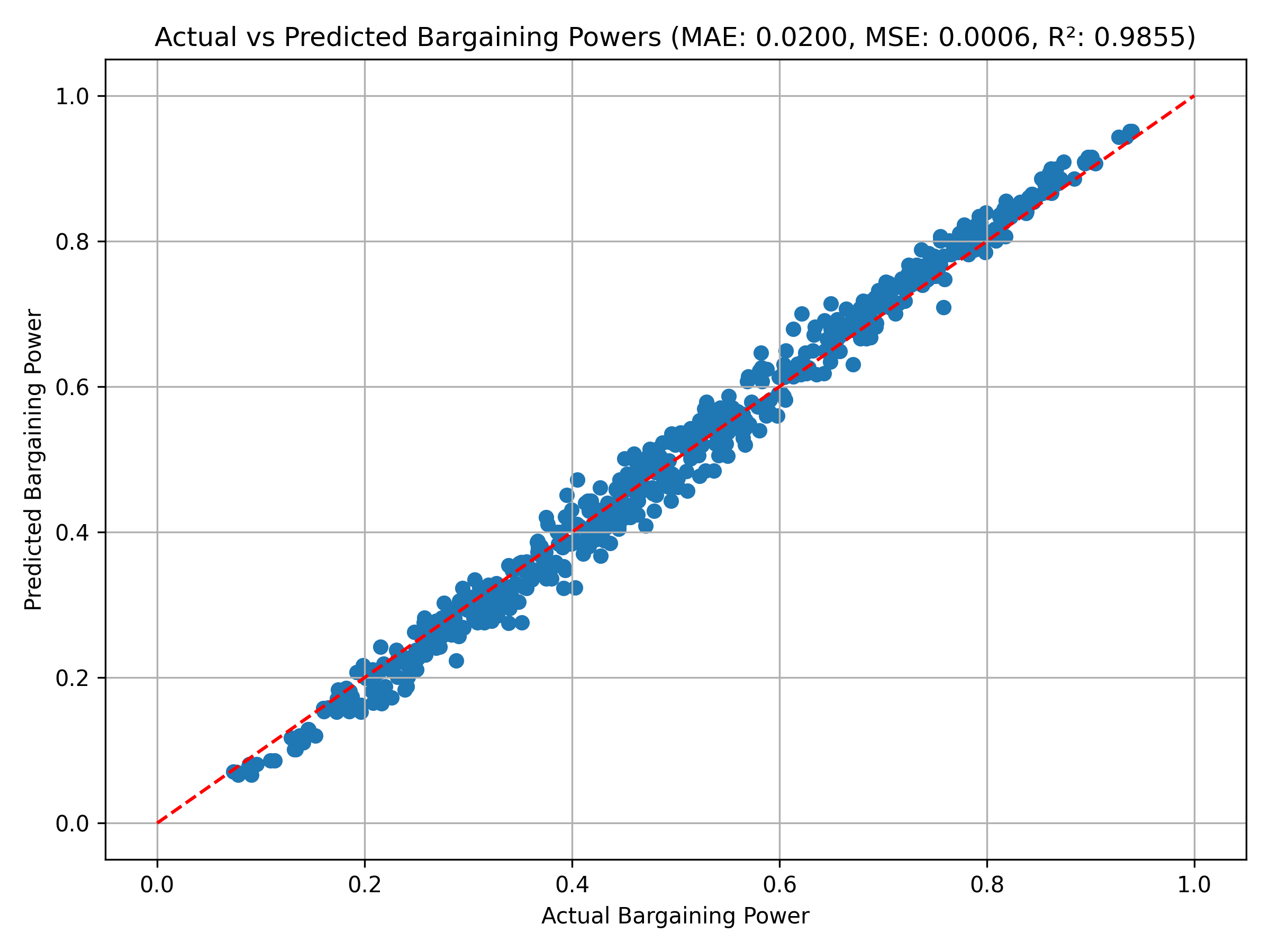}
         \caption{Bargaining Power}
         \label{fig:pi_dense}
     \end{subfigure}%
     \begin{subfigure}[b]{0.5\textwidth}
         \centering
         \includegraphics[width=\textwidth]{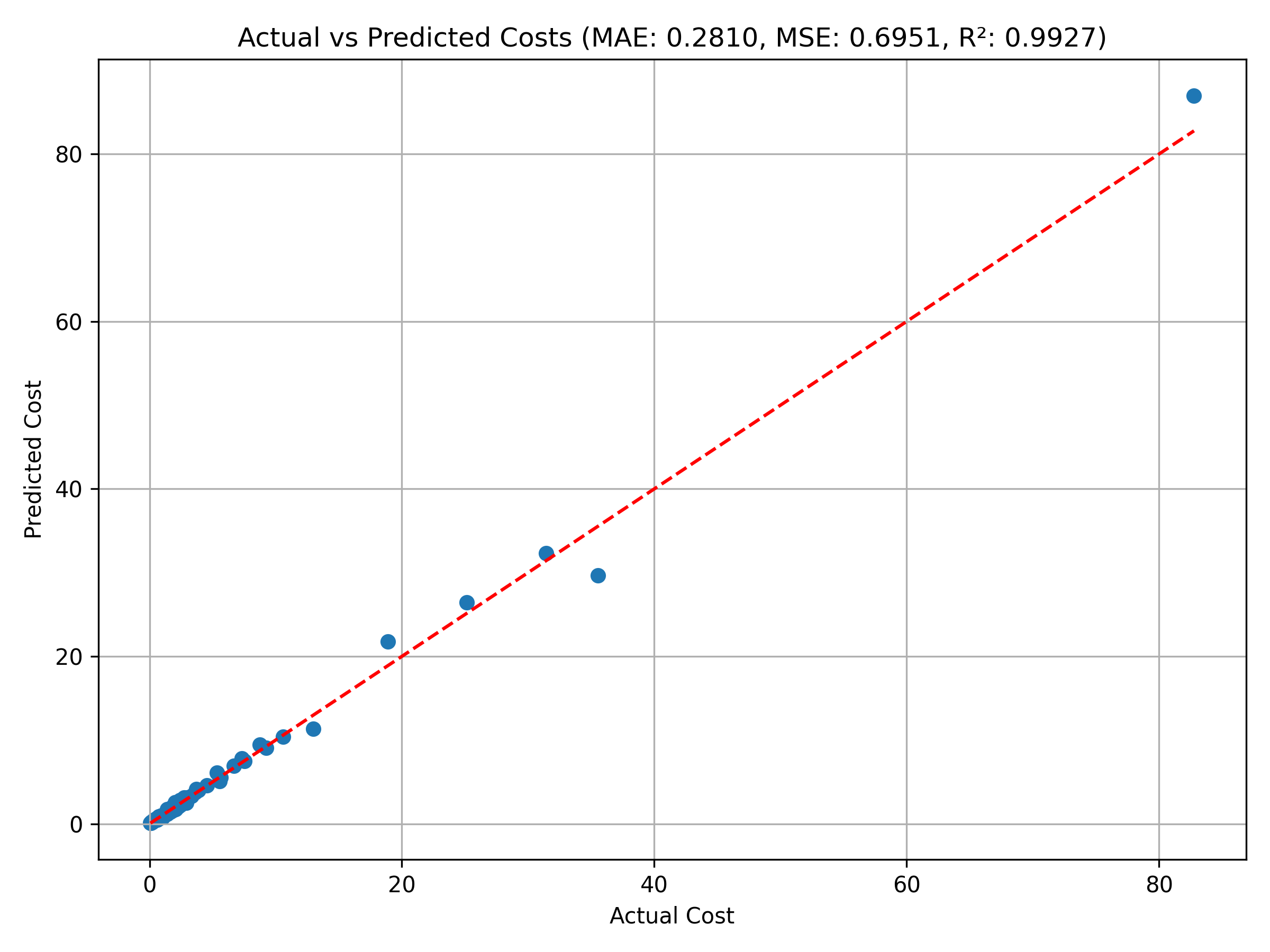}
         \caption{Holding Costs}
         \label{fig:c_dense}
     \end{subfigure}
     \begin{subfigure}[b]{0.5\textwidth}
         \centering
         \includegraphics[width=\textwidth]{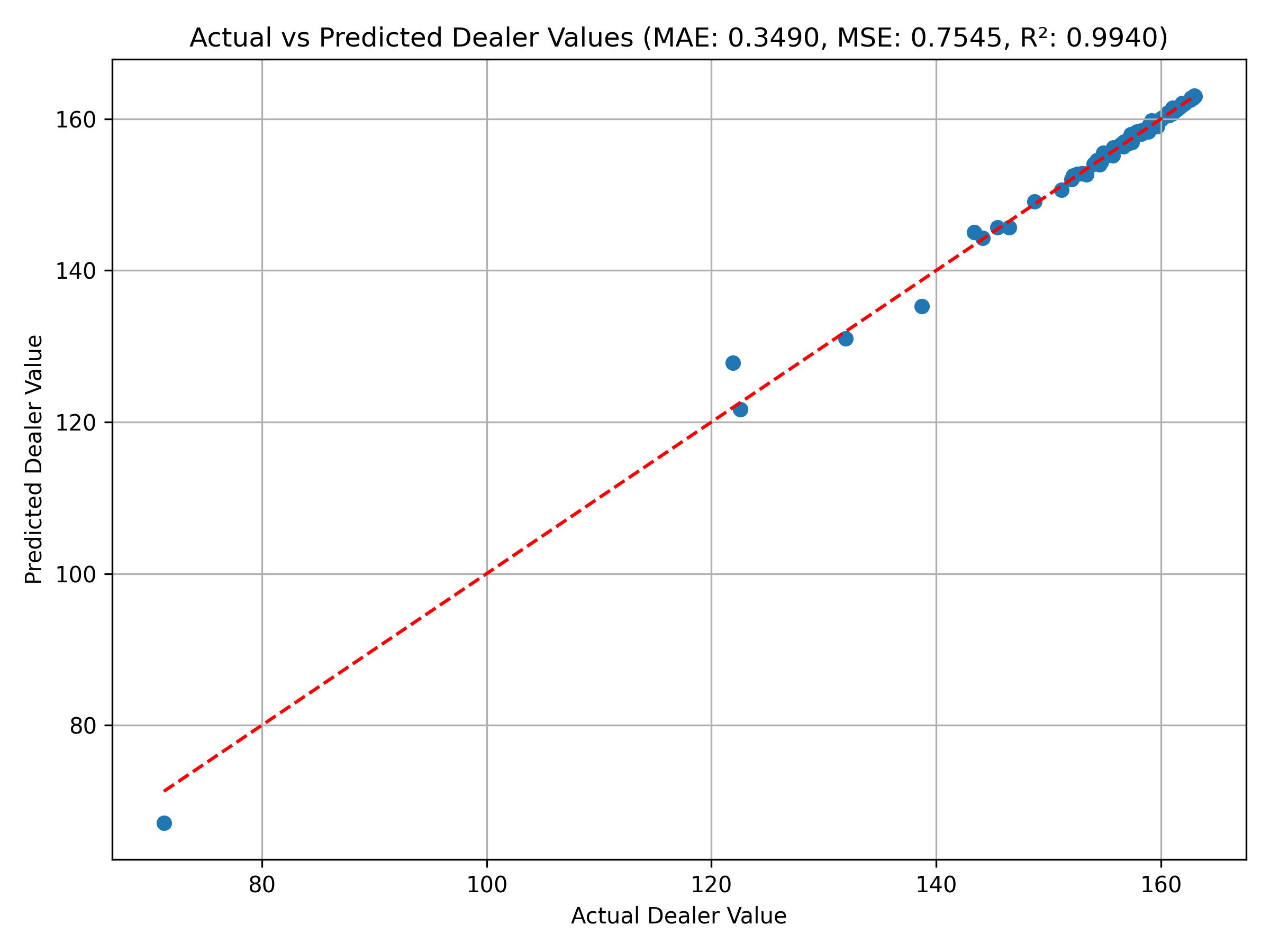}
         \caption{Dealer Values}
         \label{fig:v_dense}
     \end{subfigure}%
          \begin{subfigure}[b]{0.5\textwidth}
         \centering
         \includegraphics[width=\textwidth]{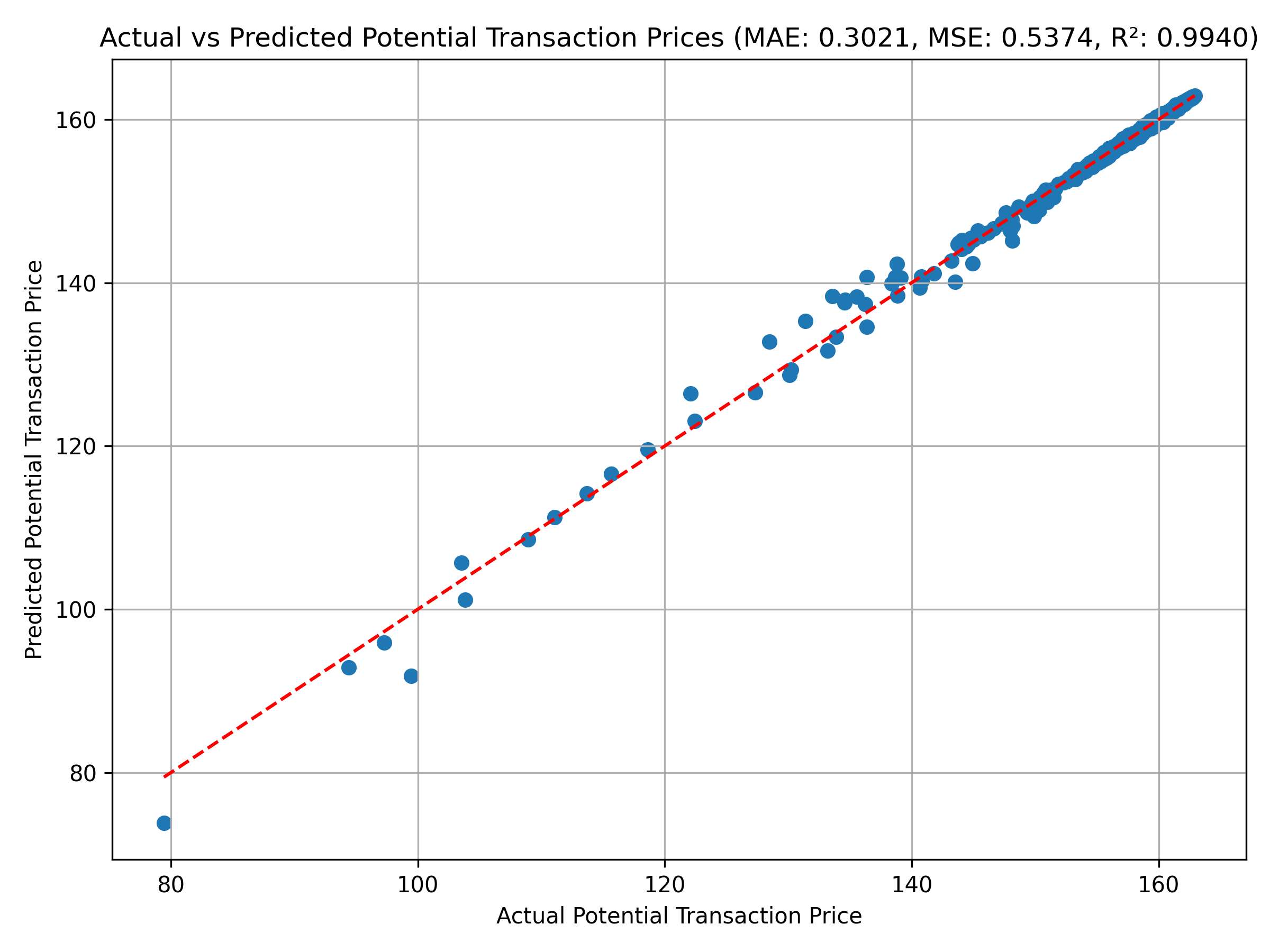}
         \caption{Potential Transaction Prices}
         \label{fig:p_dense}
     \end{subfigure}
        \caption{Comparison of Predicted vs. Actual Latent Variables in Dense Random Networks}
        \label{fig:est_latent_dense}
\end{figure}

\par As a comparison, we use OLS regression to estimate the following equation 
\begin{equation}\label{eq:OLS}
    p_{ijkt} =  X_{kt} \beta_x + Y_{it} \beta_{s} + Y_{jt} \beta_{b} + E_{ijt} \eta  + \delta_{ijkt}
\end{equation}
where $p_{ijkt}$ are observed prices, $Y_{it}$ are seller $i$'s features, $Y_{jt}$ are buyer $j$'s features, and $E_{ijt}$ are features of dealer $i$ and $j$'s relationship, $\delta_{ijkt}$ is the error term. Specifications similar to equation (\ref{eq:OLS}) are commonly used in the literature to estimate the impact of dealer features or asset features on asset prices. 
\par We can add various centrality measures of the dealer buyer and dealer seller to control for the network structure,
\begin{equation}\label{eq:OLS_centrality}
    p_{ijkt} = X_{kt} \beta_x + Y_{it} \beta_{s} + Y_{jt} \beta_{b} + E_{ijt} \eta + C_{ikt} \gamma_{s} + C_{jkt} \gamma_{b} + \delta_{ijkt}.
\end{equation}
where $C_{ikt}$ and $C_{jkt}$ are centrality measures for seller $i$ and buyer $j$ respectively. These centrality measures include degree centrality, eigenvector centrality, and betweenness centrality.  
\par Besides, we can account for the heterogeneous impact of these features for dealers in different network positions with the following regression, 
\begin{eqnarray}\label{eq:OLS_centrality_interactions}
 &&   p_{ijkt} = X_{kt} \beta_x + Y_{it} \beta_{s} + Y_{jt} \beta_{b} + E_{ijt} \eta + C_{ikt} \gamma_{s} + C_{jkt} \gamma_{b} + \boldsymbol{I}_{ikt}'\boldsymbol{\zeta}_{s} + \boldsymbol{I}_{jkt}' \boldsymbol{\zeta}_b + \delta_{ijkt}.
\end{eqnarray}
where $\boldsymbol{I}_{ikt}$ are interaction terms between the asset features, trader features, relationship features and the seller $i$'s centrality measures, and $\boldsymbol{I}_{jkt}$ are interaction terms between the asset features, trader features, relationship features and the buyer $j$'s centrality measures.

\par Table \ref{tab:model_comparison_dense} shows the performance comparison between OLS and TGNN. Compared with OLS, TGNN has the highest $R^2$, the lowest Mean Absolute Error and Mean Squared Error. Figure \ref{fig:prediction_comparison_dense} shows the predicted prices against actual prices.  We can see that TGNN's predicted prices fit the data better than the best OLS estimation with all centrality interaction terms. Table \ref{tab:model_comparison_dense} also shows that TGNN has the lowest Akaike Information Criterion (AIC) and Bayesian Information Criterion (BIC), suggesting that TGNN is the best statistical model for recovering the parameters in dense random networks.\footnote{The Akaike Information Criterion (AIC) and the Bayesian Information Criterion (BIC) are model selection tools that balance model fit and complexity. Both criteria penalize the likelihood function based on the number of parameters to avoid overfitting. AIC is defined as \( \text{AIC} = 2k - 2\ln(L) \), where \( k \) is the number of estimated parameters and \( L \) is the maximized value of the likelihood function. BIC is given by \( \text{BIC} = \ln(n)k - 2\ln(L) \), where \( n \) is the number of observations. While both criteria favor models with lower values, BIC imposes a heavier penalty on model complexity, especially when the sample size is large. As a result, BIC tends to select more parsimonious models compared to AIC.}  
\par TGNN performs better than OLS regressions for two reasons. First, it structurally incorporates the network structures which can improve the estimation compared with centrality measures that are not micro-founded. Second, TGNN utilizes customer values that serve as outside options for the dealer sellers. These outside options never appear in the interdealer transaction prices, but serve as a threshold to truncate the data of observable prices. OLS regressions without accounting for these customer values may lead to biased estimation and lower explanatory power. 
\begin{table}[htbp]
\caption{Performance Comparison: OLS vs. TGNN Models in Dense Random Networks}
\label{tab:model_comparison_dense}
\begin{center}
    \begin{tabular}{lcccccc}
\hline
Model & $R^2$ & MAE & MSE & Parameters & AIC & BIC \\ \hline
OLS Basic & 0.3486 & 1.5659 & 3.9486 & 5 & 103.4 & 114.5 \\ 
OLS + Degree & 0.3967 & 1.5310 & 3.6570 & 9 & 106.2 & 126.1 \\ 
OLS + Eigenvector & 0.3853 & 1.4249 & 3.7259 & 7 & 103.4 & 119.0 \\ 
OLS + Betweenness & 0.3552 & 1.5889 & 3.9083 & 7 & 106.7 & 122.2 \\ 
OLS + All Centrality & 0.4551 & 1.3934 & 3.3030 & 13 & 107.2 & 136.1 \\ 
OLS + Eigenvector Interactions & 0.4841 & 1.2332 & 3.1273 & 15 & 107.5 & 140.8 \\ 
OLS + Centrality Interactions & 0.8151 & 0.7490 & 1.1211 & 45 & 97.8 & 197.6 \\ 
TGNN & 0.9930 & 0.1548 & 0.0427 & 3 & -208.4 & -201.8 \\ 
\hline
\end{tabular}
\end{center}
    \footnotesize{\textit{Note:} This table shows the performance comparison of OLS regressions and TGNN in dense random networks. OLS Basic refers to the regression in equation (\ref{eq:OLS}). OLS + Degree, OLS + Eigenvector, OLS + Betweenness and OLS + All Centrality refer to the regressions in equation (\ref{eq:OLS_centrality}) with degree centrality, eigenvector centrality, betweenness centrality, and all the above centrality measures. OLS + Eigenvector Interactions refers to the regression equation (\ref{eq:OLS_centrality_interactions}) with eigenvector centrality and its interaction terms with all the asset features, dealer-buyer features, dealer-seller features and relationship features. OLS + Centrality Interactions refers to the regression equation (\ref{eq:OLS_centrality_interactions}) with all centrality measures and its interaction terms with all the asset features, dealer-buyer features, dealer-seller features and relationship features. Trading Graph Neural Network (TGNN) is the structural estimation method described by Algorithm \ref{alg:train} and \ref{alg:CI}. }
\end{table}

\begin{figure}[htbp]
\begin{center}
    \includegraphics[width=0.7\linewidth]{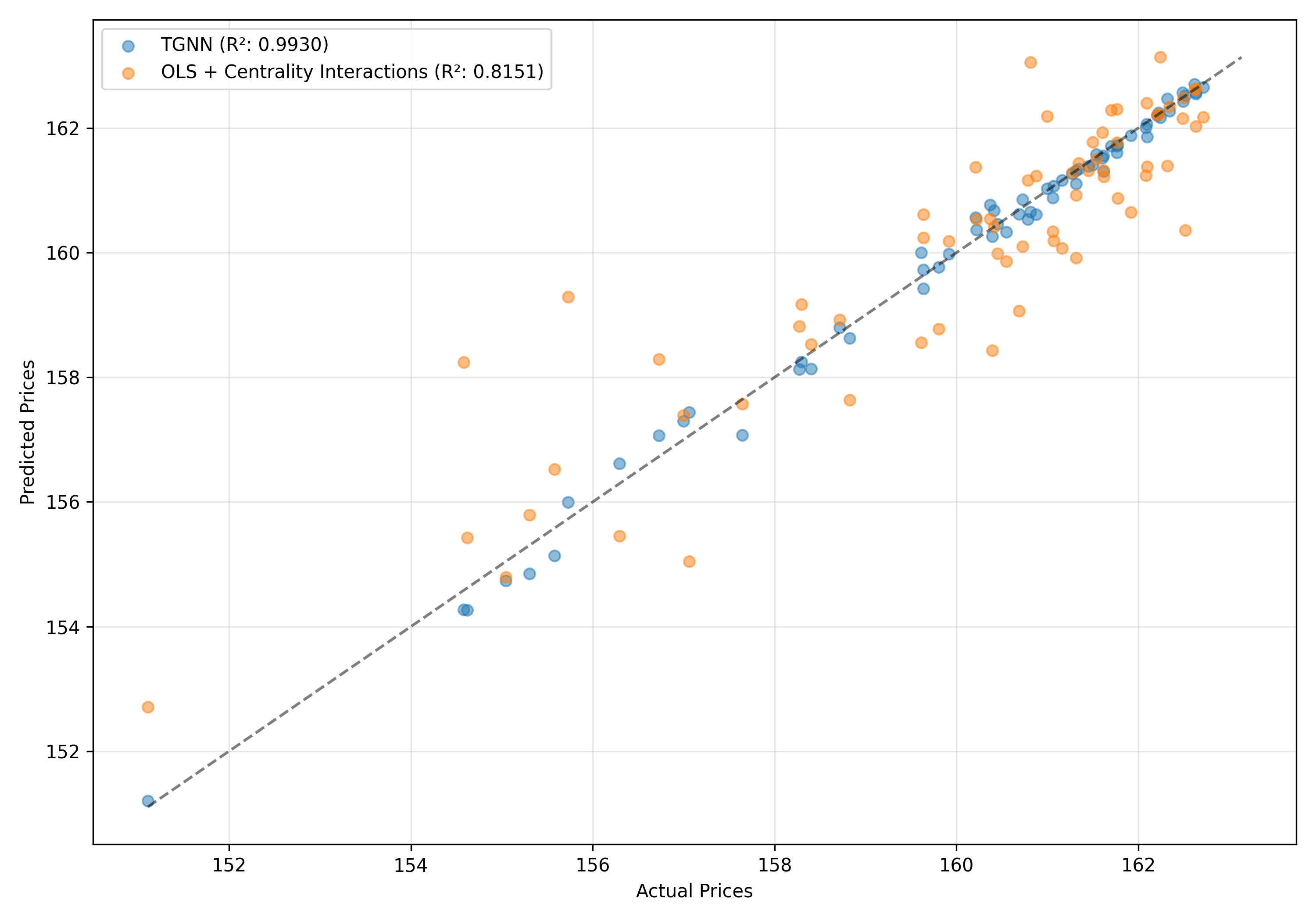}
\end{center}
    \caption{Prediction Comparison: OLS with Centrality Interactions vs. TGNN in Dense Random Networks}
    \label{fig:prediction_comparison_dense}
         \footnotesize{\textit{Note:} This figure shows the comparison of accuracy of predicted prices between the OLS with the highest $R^2$ -- OLS with centrality interaction and TGNN. The dash line is the reference line where the predicted prices are equal to the actual prices. }
\end{figure}

\subsection{Sparse Networks} 
\par Previously, we have tested the performance of TGNN with dense Erdos–Rényi random networks, we further examine the performance in sparser networks. The data generation process is the same as the previous test case, but each potential edge is included in the network with 20\% probability. Figure \ref{fig:trading_networks_sparse} shows the simulated sparse network. Table \ref{tab:summary_stats_sparse} summarizes the statistics of the variables.

\begin{figure}[htbp]
    \includegraphics[width=\linewidth]{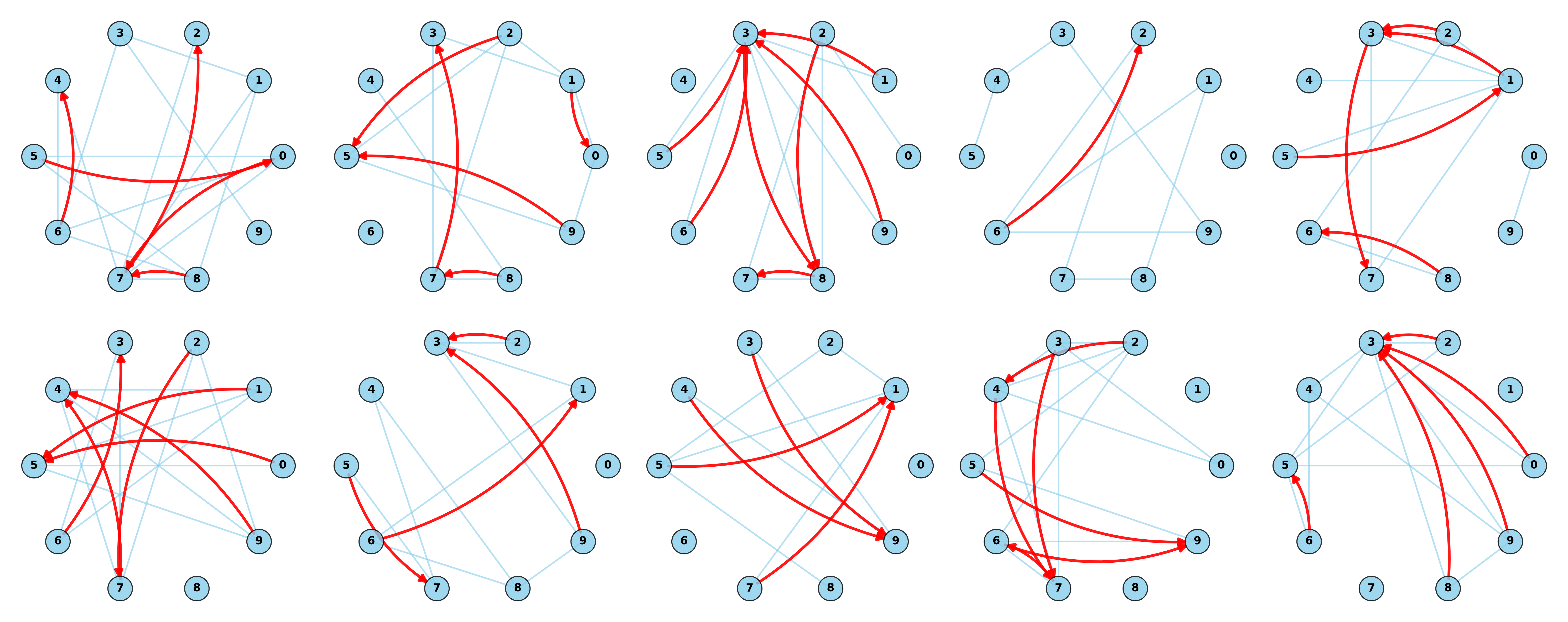}
    \caption{The Structure of Sparse Random Networks}
    \label{fig:trading_networks_sparse}
     \footnotesize{\textit{Note:} This figure shows the structure of interdealer networks generated for 10 dealers, 2 assets over 5 days with 206 edges. The blue lines indicate the directed edges, where the edge is directed from seller to buyer. The red arrows denote observable trades, with the arrowheads indicating the direction toward the buyer. \raggedright}
\end{figure}   
\begin{table}[htbp]
\caption{Summary Statistics of Trading Network Variables in Sparse Random Networks}
\begin{center}
    \begin{tabular}{lccccc}
\hline
Variable & N & Min & Max & Mean & Std \\ \hline
\textit{Observable Variables} \\
\hline
Asset Feature X & 100 & -1.1229 & 2.2082 & 0.1920 & 0.8205 \\ 
Dealer Feature Y & 100 & -2.2064 & 2.8140 & 0.0969 & 0.9887 \\ 
Relationship Feature E & 206 & -2.3768 & 2.7421 & 0.0251 & 0.9540 \\ 
Observed Best Prices & 48 & 152.1473 & 163.1463 & 158.6704 & 2.7209 \\ 
Customer Values & 100 & 148.5162 & 163.7988 & 156.2434 & 4.3160 \\
\hline
\textit{Unobservable Latent Variables} \\
\hline
Dealer Values & 100 & 74.7297 & 163.5353 & 154.5093 & 10.8971 \\ 
Bargaining Powers & 206 & 0.0781 & 0.9406 & 0.5028 & 0.2027 \\ 
Potential Transaction Prices & 206 & 87.8948 & 163.1463 & 155.3911 & 8.9097 \\ 
Costs & 100 & 0.0750 & 85.9399 & 3.9459 & 10.1757 \\  
\hline
\end{tabular}
\end{center}
\label{tab:summary_stats_sparse}
\footnotesize{\textit{Note:}  This table shows the generated data of observable and latent variables in the test case of sparse random networks. } 
\end{table}
\par Table \ref{tab:param_estimates_sparse} shows the parameter estimates with confidence intervals from TGNN in sparse networks. It shows that the estimates are close to the true values.\footnote{ Appendix Figure \ref{fig:est_sparse} shows the bootstrap distribution of the estimators in sparse random networks. Appendix Figure \ref{fig:est_latent_sparse} shows the comparison of predicted versus the actual bargaining power, holding costs, dealer values and potential transaction costs in sparse random networks. }
\begin{table}[htbp]
\caption{Parameter Estimates with Bootstrap Confidence Intervals in Sparse Random Networks}
\begin{center}
    \begin{tabular}{lcccc}
\hline
Parameter & True Value & Estimate & 95\% CI  \\ \hline
$\hat{\beta}_x$ & 1.0000 & 1.0345 & [0.9775, 1.1113] \\
$\hat{\beta}_y$ & 1.0000 & 0.9946 & [0.9396, 1.0394] \\
$\hat{\eta}$ & 1.0000 & 1.0232 & [0.9414, 1.1244] \\
\hline
\end{tabular}
\end{center}
\label{tab:param_estimates_sparse}
 \footnotesize{\textit{Note:}  This table shows the estimated parameters on asset features $\hat{\beta}_x$, on dealer features $\hat{\beta}_y$ and relationship features $\hat{\eta}$ with 95\% confidence intervals using TGNN in sparse random networks. } 
\end{table}

\par Table \ref{tab:model_comparison_sparse} shows the comparison between TGNN and OLS with different specifications. We can see that the TGNN has the highest $R^2$, lowest MAE and MSE among all specifications. AIC and BIC suggest that TGNN remains the best model for recovering the parameters in sparse random networks.
\begin{table}[htbp]
\caption{Performance Comparison: OLS vs. TGNN Models in Sparse Random Networks}
\begin{center}
\begin{tabular}{lcccccc}
\hline
Model & $R^2$ & MAE & MSE & Parameters & AIC & BIC \\ \hline
OLS Basic & 0.5689 & 1.4524 & 3.1914 & 5 & 65.7 & 75.1 \\ 
OLS + Degree & 0.5773 & 1.4331 & 3.1291 & 9 & 72.8 & 89.6 \\ 
OLS + Eigenvector & 0.5909 & 1.3756 & 3.0285 & 7 & 67.2 & 80.3 \\ 
OLS + Betweenness & 0.5732 & 1.4317 & 3.1597 & 7 & 69.2 & 82.3 \\ 
OLS + All Centrality & 0.5964 & 1.4061 & 2.9881 & 13 & 78.5 & 102.9 \\ 
OLS + Eigenvector Interactions & 0.6794 & 1.1939 & 2.3731 & 15 & 71.5 & 99.6 \\ 
OLS + Centrality Interactions & 0.8913 & 0.5701 & 0.8046 & 45 & 79.6 & 163.8 \\ 
TGNN & 0.9906 & 0.1938 & 0.0695 & 3 & -122.0 & -116.4 \\ 
\hline
\end{tabular}
\end{center}
\label{tab:model_comparison_sparse}
    \footnotesize{\textit{Note:}  This table shows the performance comparison of OLS regressions and TGNN in sparse random networks. OLS Basic refers to the regression in equation (\ref{eq:OLS}). OLS + Degree, OLS + Eigenvector, OLS + Betweenness and OLS + All Centrality refer to the regressions in equation (\ref{eq:OLS_centrality}) with degree centrality, eigenvector centrality, betweenness centrality, and all the above centrality measures. OLS + Eigenvector Interactions refers to the regression equation (\ref{eq:OLS_centrality_interactions}) with eigenvector centrality and its interaction terms with all the asset features, dealer-buyer features, dealer-seller features and relationship features. OLS + Centrality Interactions refers to the regression equation (\ref{eq:OLS_centrality_interactions}) with all centrality measures and its interaction terms with all the asset features, dealer-buyer features, dealer-seller features and relationship features. Trading Graph Neural Network (TGNN) is the structural estimation method described by Algorithm \ref{alg:train} and \ref{alg:CI}. }
\end{table}

\subsection{Core-periphery Networks} 
So far, we've examined TGNN's performance in test cases of Erdos–Rényi (ER) random graphs. We further examine the performance of TGNN with core-periphery networks. A \textit{core-periphery financial network} is a market structure where financial institutions, assets, or trading venues are divided into a \textit{core} and a \textit{periphery}, based on their connectivity, liquidity, and influence in trading activities. Core traders are highly interconnected, serve as central hubs for trading and price discovery, and typically include major exchanges, large banks, or dominant market makers. In contrast, periphery traders are less connected. Most OTC markets are core-periphery trading networks, e.g. bond market and CDS market \citep{wang2016core}.
\par We generate core-periphery networks with 4 core traders, 16 periphery traders, 2 assets, and 5 days. The probability for a core dealer to have a link with another core dealer is 90\%, for a core dealer to have a link with a periphery dealer is 70\%, and for a periphery dealer to have a link with another periphery dealer is 1\%. Figure  \ref{fig:core_periphery_network} shows the simulated core-periphery networks. Table \ref{tab:summary_stats_cp} shows the summary of statistics. 

\begin{figure}[htbp]
    \includegraphics[width=\linewidth]{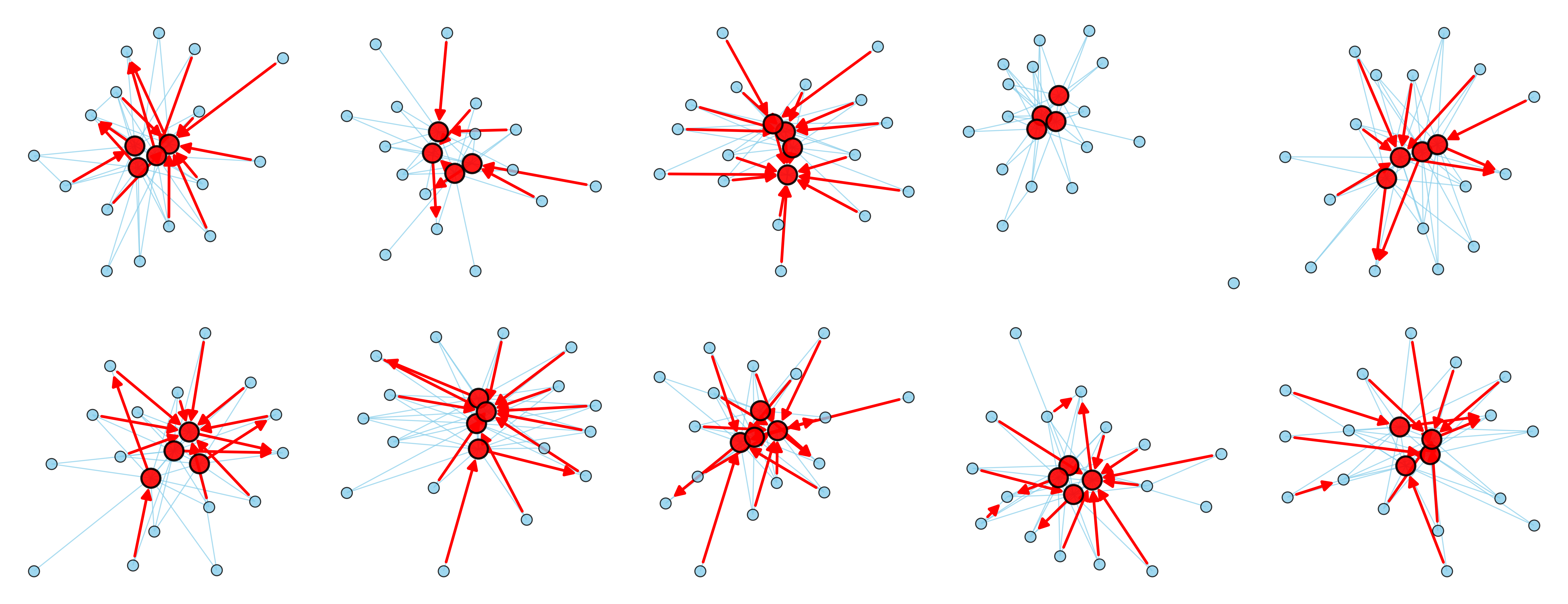}
    \caption{The Structure of Core-Periphery Networks}
    \label{fig:core_periphery_network}
     \footnotesize{\textit{Note:} This figure shows the structure of interdealer networks generated for 4 core dealers (red node), 16 periphery dealers (blue node), 2 assets over 5 days with 998 edges. The blue lines indicate the directed edges, where the edge is directed from seller to buyer. The red arrows denote observable trades, with the arrowheads indicating the direction toward the buyer. \raggedright}
\end{figure}   

\begin{table}[htbp]
\caption{Summary Statistics of Trading Network Variables in Core-periphery Networks}
\begin{center}
    \begin{tabular}{lccccc}
\hline
Variable & N & Min & Max & Mean & Std \\ \hline
\textit{Observable Variables}\\
\hline
Asset Feature X & 200 & -1.1229 & 2.2082 & 0.1920 & 0.8205 \\ 
Dealer Feature Y & 200 & -2.2064 & 2.8140 & 0.1205 & 0.9103 \\ 
Relationship Feature E & 998 & -3.0668 & 2.9246 & -0.0366 & 0.9494 \\ 
Observed Prices & 123 & 152.9716 & 163.7037 & 160.0051 & 2.0540 \\ 
Customer Values & 200 & 148.6951 & 163.9211 & 156.1325 & 4.1576 \\ 
\hline
\textit{Unobservable Latent Variables}\\
\hline
Dealer Values & 200 & 21.5183 & 163.7215 & 155.8367 & 11.6800 \\ 
Bargaining Powers & 998 & 0.0416 & 0.9546 & 0.4933 & 0.2001 \\ 
Potential Transaction Prices & 998 & 41.8713 & 163.7037 & 155.3923 & 12.4922 \\ 
Costs & 200 & 0.0362 & 139.2609 & 3.8084 & 11.0124 \\ 
\hline
\end{tabular}
\end{center}
\label{tab:summary_stats_cp}
 \footnotesize{\textit{Note:}  This table shows the generated data of observable and latent variables in the test case of core-periphery networks. } 
\end{table}

\par Table \ref{tab:param_estimates_cp} shows the estimated parameters with confidence intervals from TGNN in core-periphery networks. We can see that TGNN recovers estimates that are close to the true values.\footnote{ Appendix Figure \ref{fig:est_cp} shows the bootstrap distribution of the estimators in core-periphery networks. Appendix Figure \ref{fig:est_latent_cp} shows the comparison of predicted versus the actual bargaining power, holding costs, dealer values and potential transaction costs in core-periphery networks. }
\begin{table}[htbp]
\caption{Parameter Estimates with Bootstrap Confidence Intervals in Core-Periphery Networks}
\begin{center}
    \begin{tabular}{lcccc}
\hline
Parameter & True Value & Estimate & 95\% CI  \\ \hline
$\hat{\beta}_x$ & 1.0000 & 1.0041 & [0.9480, 1.0459] \\
$\hat{\beta}_y$ & 1.0000 & 0.9703 & [0.9138, 1.0376] \\
$\hat{\eta}$ & 1.0000 & 0.9934 & [0.9158, 1.0348] \\
\hline
\end{tabular}
\end{center}
\label{tab:param_estimates_cp}
   \footnotesize{\textit{Note:}  This table shows the estimated parameters on asset features $\hat{\beta}_x$, on dealer features $\hat{\beta}_y$ and relationship features $\hat{\eta}$ with 95\% confidence intervals using TGNN in core-periphery networks. } 
\end{table}

\par Table \ref{tab:model_comparison_cp} shows the comparison between TGNN and OLS with centrality measures and interaction terms. We can see that TGNN still has the highest $R^2$, lowest MAE and MSE. AIC and BIC suggest that TGNN is still the best among these statistical models for recovering parameters in core-periphery networks. 

\begin{table}[htbp]
\caption{Performance Comparison: OLS vs. TGNN Models in Core-periphery Networks}
\begin{center}
\begin{tabular}{lcccccc}
\hline
Model & $R^2$ & MAE & MSE & Parameters & AIC & BIC \\ \hline
OLS Basic & 0.5893 & 1.0108 & 1.7329 & 5 & 77.6 & 91.7 \\ 
OLS + Degree & 0.6262 & 0.9647 & 1.5771 & 9 & 74.0 & 99.3 \\ 
OLS + Eigenvector & 0.6041 & 0.9730 & 1.6702 & 7 & 77.1 & 96.8 \\ 
OLS + Betweenness & 0.6146 & 0.9661 & 1.6260 & 7 & 73.8 & 93.5 \\ 
OLS + All Centrality & 0.6936 & 0.8497 & 1.2928 & 13 & 57.6 & 94.1 \\ 
OLS + Eigenvector Interactions & 0.7125 & 0.7933 & 1.2131 & 15 & 53.8 & 95.9 \\ 
OLS + Centrality Interactions & 0.8188 & 0.6578 & 0.7647 & 45 & 57.0 & 183.5 \\ 
TGNN & 0.9843 & 0.1815 & 0.0662 & 3 & -328.0 & -319.6 \\ 
\hline
\end{tabular}
\end{center}
\label{tab:model_comparison_cp}
    \footnotesize{\textit{Note:}  This table shows the performance comparison of OLS regressions and TGNN in core-periphery networks. OLS Basic refers to the regression in equation (\ref{eq:OLS}). OLS + Degree, OLS + Eigenvector, OLS + Betweenness and OLS + All Centrality refer to the regressions in equation (\ref{eq:OLS_centrality}) with degree centrality, eigenvector centrality, betweenness centrality, and all the above centrality measures. OLS + Eigenvector Interactions refers to the regression equation (\ref{eq:OLS_centrality_interactions}) with eigenvector centrality and its interaction terms with all the asset features, dealer-buyer features, dealer-seller features and relationship features. OLS + Centrality Interactions refers to the regression equation (\ref{eq:OLS_centrality_interactions}) with all centrality measures and its interaction terms with all the asset features, dealer-buyer features, dealer-seller features and relationship features. Trading Graph Neural Network (TGNN) is the structural estimation method described by Algorithm \ref{alg:train} and \ref{alg:CI}. }
\end{table}

\section{Applications}

\par TGNN offers powerful applications for analyzing decentralized markets in economics and finance. From a methodological perspective, the TGNN enables structural estimation of economic models directly from observed transaction data. This approach bridges the gap between theoretical models of OTC markets and empirical analysis, allowing researchers to test economic theories while accounting for network effects. It accommodates heterogeneity in trading relationships, allowing for varied bargaining powers across different participant pairs and contextual factors that influence transaction outcomes. This flexibility makes the model adaptable to diverse market settings while maintaining its economic interpretability, where parameters directly correspond to economic quantities of interest rather than abstract neural network weights. The model estimates can be used for counterfactual analysis, such as simulating market outcomes under different regulatory regimes, market structures, or entry/exit scenarios.

\par A significant application lies in over-the-counter (OTC) markets, where trades occur through bilateral negotiations rather than central exchanges, e.g., fixed income securities, municipal bonds, and interbank lending networks. TGNN can quantify dealer relationships and identify key market participants. The model's ability to estimate marginal holding costs (\(c_i\)) and bargaining powers (\(\alpha_{ij}\)) provides crucial insights into price formation mechanisms and market efficiency. This is particularly valuable for regulatory oversight and systemic risk assessment. By recovering the fundamental parameters that drive trading decisions, the model helps identify potential market manipulation, or concentration of market power. Moreover, the network structure revealed by the model can highlight vulnerable nodes in the financial system that might propagate shocks during periods of market stress, contributing to macroprudential policy and financial stability analysis.

\par Beyond traditional finance, this methodology extends to emerging decentralized markets, including cryptocurrency trading and peer-to-peer lending platforms. As these markets mature, understanding their network structure and price formation becomes increasingly important for investors, platform designers, and regulators. 

\clearpage
\bibliography{ref}

\clearpage
\appendices

\section{Proof}

\begin{proof}[Proof of Theorem \ref{thm:eqlm}]
We show that the model defined by equations (\ref{eq:value}) and (\ref{eq:price}) has a unique equilibrium for the vector of dealer values $\{v_{ikt}\}_i$ given $\{c_{ikt},u_{ikt}\}_{i,k,t}$ and $\{\pi_{ijkt}\}_{i,j,k,t}$ at time $t$ and asset $k$. Fix the time $t$ and asset $k$, and suppress these indices to simplify notation: write $v_i$ for $v_{ikt}$, $c_i$ for $c_{ikt}$, $u_i$ for $u_{ikt}$, and $p_{ij}$ for $p_{ijkt}$. Define the vector $\boldsymbol{v} = (v_1, \dots, v_n) \in \mathbb{R}^n$.
\par Define the mapping $T: \mathbb{R}^n \to \mathbb{R}^n$ such that the $i^{th}$ component of $T(\boldsymbol{v}): \mathbb{R}^n \to \mathbb{R}$ is given by
\begin{equation*}
    T_i(\boldsymbol{v}) = -c_i + \max\left\{ \max_{j \in \mathcal{N}(i)} \left\{ \pi_{ij} v_i + (1 - \pi_{ij}) v_j \right\}, u_i \right\}.
\end{equation*}
\par We want to show that $T$ is a contraction mapping. Suppose all $\pi_{ij} \in [\varepsilon, 1 - \varepsilon]$ for some $\varepsilon > 0$.

\par Let $\boldsymbol{v}, \boldsymbol{v}' \in \mathbb{R}^n$. For each $i$, define
\[
    p_{ij} = \pi_{ij} v_i + (1 - \pi_{ij}) v_j, \quad p_{ij}' = \pi_{ij} v_i' + (1 - \pi_{ij}) v_j'.
\]
\par Then for any $j\in \mathcal{N}(i)$,
\[
    |p_{ij} - p_{ij}'| 
    \leq \max\{\pi_{ij}, 1 - \pi_{ij}\} \cdot \|\boldsymbol{v} - \boldsymbol{v}'\|_\infty 
    \leq (1 - \varepsilon)\|\boldsymbol{v} - \boldsymbol{v}'\|_\infty.
\]
\par Since $T_i(\boldsymbol{v})$ takes a maximum over $u_i$ and the $\{p_{ij}\}_{j\in \mathcal{N}(i)}$ terms, we have
\[
    |T_i(\boldsymbol{v}) - T_i(\boldsymbol{v}')| 
    \leq \max_{j \in \mathcal{N}(i)} |p_{ij} - p_{ij}'| 
    \leq \|\boldsymbol{v} - \boldsymbol{v}'\|_\infty.
\]
\par So,
\[
    \|T(\boldsymbol{v}) - T(\boldsymbol{v}')\|_\infty 
    = \max_i |T_i(\boldsymbol{v}) - T_i(\boldsymbol{v}')|
    \leq (1 - \varepsilon) \|\boldsymbol{v} - \boldsymbol{v}'\|_\infty.
\]

\par Hence, $T$ is a contraction mapping on the complete metric space $(\mathbb{R}^n, \|\cdot\|_\infty)$.
\par By the Banach fixed point theorem, $T$ has a unique fixed point. That is, there exists a unique vector $\boldsymbol{v}^* \in \mathbb{R}^n$ such that $T(\boldsymbol{v}^*) = \boldsymbol{v}^*$. This proves the model has a unique equilibrium. Furthermore, we can find the equilibrium $\boldsymbol{v}^*$ by starting with an arbitrary element  $\boldsymbol{v}^*\in\mathbb{R}^n$, and define a sequence $\{\boldsymbol{v}^{(\ell)}\}_{\ell\in \mathbb{N}}$ by $\boldsymbol{v}^{(\ell)}=T(\boldsymbol{v}^{(\ell-1)})$ for $n\geq 1$. Then $\lim_{\ell\rightarrow \infty} \boldsymbol{v}^{(\ell)} = \boldsymbol{v}^*$.
\par It is easy to see that $v_{min}^*\equiv\min_i\{v_{ikt}\}\geq \min_i\{u_{ikt}-c_{ikt}\}$ as the lowest outside option is $\min_i\{u_{ikt}-c_{ikt}\}$. Let $v_{max}^*\equiv\max_i\{{v}_{ikt}\}$. As $v_{max}^*\geq \max \{p_{ijkt}\}_{i,j}$ and $c_{ikt}>0$, it's impossible for $v_{max}=\max_i\{-c_{ikt}+\max_{j\in\mathcal{N}(i)}\{p_{ijkt}\}_j\}$, so $v_{max}^* =  \max_i \{u_{ikt}-c_{ikt}\}$. Intuitively, the dealer buyer with the highest value does not sell to dealers and has a resale value of $ \max_i \{u_{ikt}-c_{ikt}\}$ to its customers. 
\par Given $\boldsymbol{v}^*$, seller $i$ sells the asset $k$ at price $ p_{ikt}^*=\max\left\{ \max_{j \in \mathcal{N}(i)} \left\{ \pi_{ijkt} v_{ikt}^* + (1 - \pi_{ijkt}) v_{jkt}^* \right\}, u_{ikt} \right\}$ at time $t$, with counterparty $j=\arg\max_{j \in \mathcal{N}(i)} \left\{ \pi_{ijkt} v_{ikt}^* + (1 - \pi_{ijkt}) v_{jkt}^* \right\}$ if $p_{ikt}^*>u_{ikt}$ and with customers otherwise.
\end{proof}

\clearpage
\section{Additional Figures and Tables}
\begin{figure}[htbp]
\begin{center}
       \includegraphics[width=\linewidth]{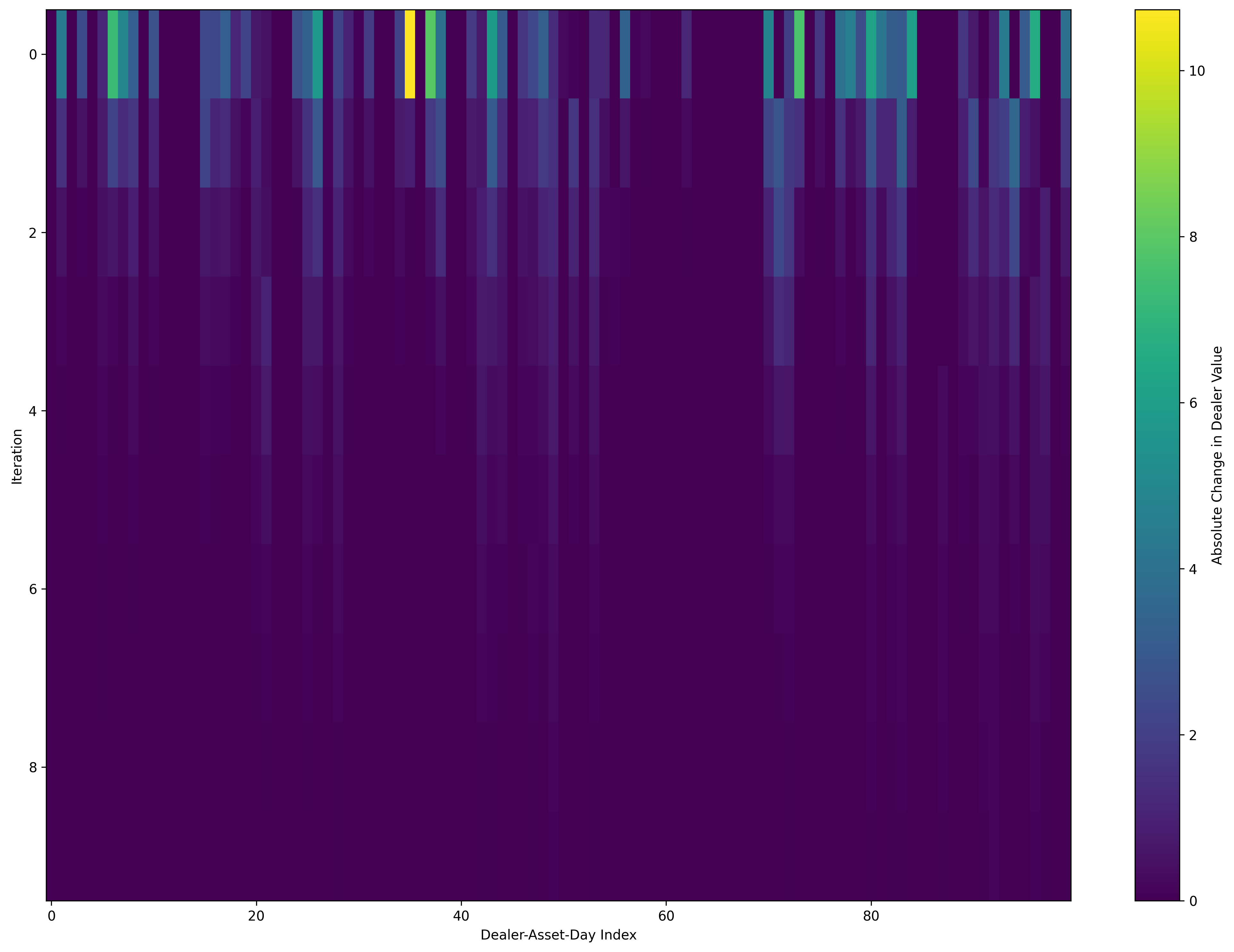}   
\end{center}
    \caption{Evolution of Dealer Values in Data Generation Process of Dense Random Networks }
    \label{fig:value_changes_dense}
\end{figure}
\footnotesize{\textit{Note:} This figure shows the evolution of dealer values with the number of interactions in the data generation of dense random networks. The horizontal axis is the index we assigned for each dealer value for each asset on each day. A darker color indicates a smaller change in dealer values between iterations.  }
\clearpage

\begin{figure}[htbp]
    \centering
    \includegraphics[width=0.6\linewidth]{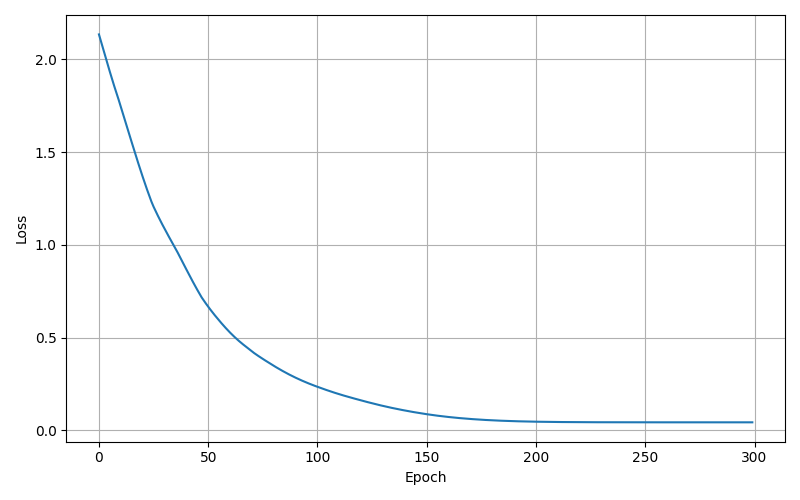}
    \caption{Training Loss in Dense Networks}
    \label{fig:test_train_loss}
\end{figure}

\begin{figure}
     \centering
     \begin{subfigure}[b]{0.6\textwidth}
         \centering
         \includegraphics[width=\textwidth]{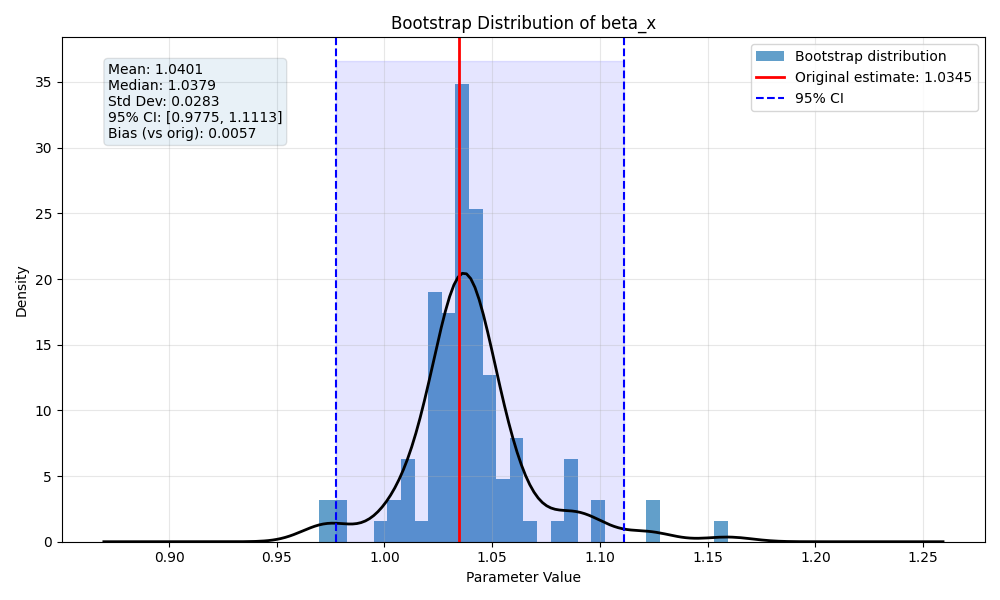}
         \caption{$\hat{\beta}_x$}
         \label{fig:beta_x_sparse}
     \end{subfigure}
     \hfill
     \begin{subfigure}[b]{0.6\textwidth}
         \centering
         \includegraphics[width=\textwidth]{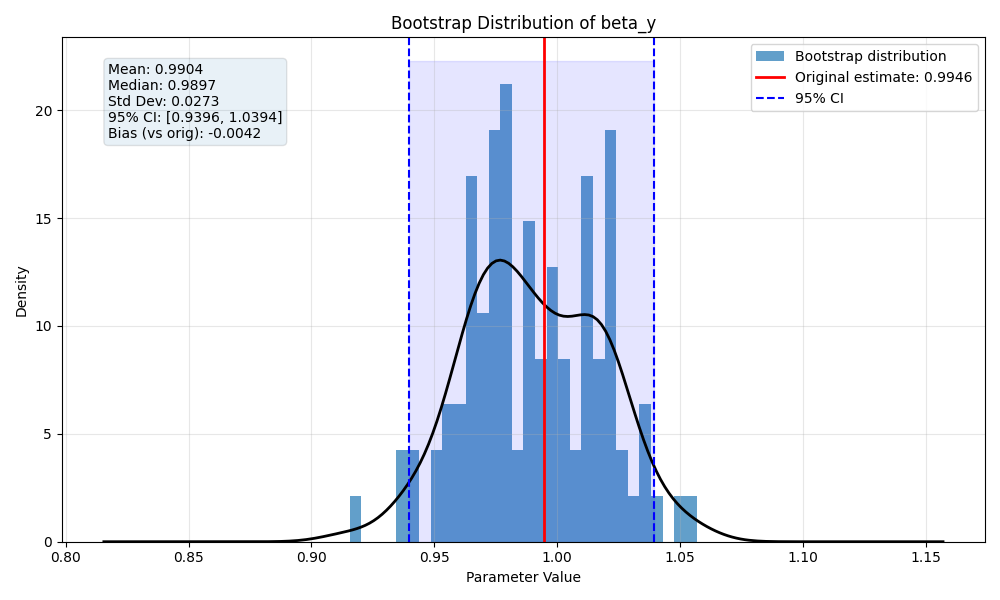}
         \caption{$\hat{\beta}_y$}
         \label{fig:beta_y_sparse}
     \end{subfigure}
     \hfill
     \begin{subfigure}[b]{0.6\textwidth}
         \centering
         \includegraphics[width=\textwidth]{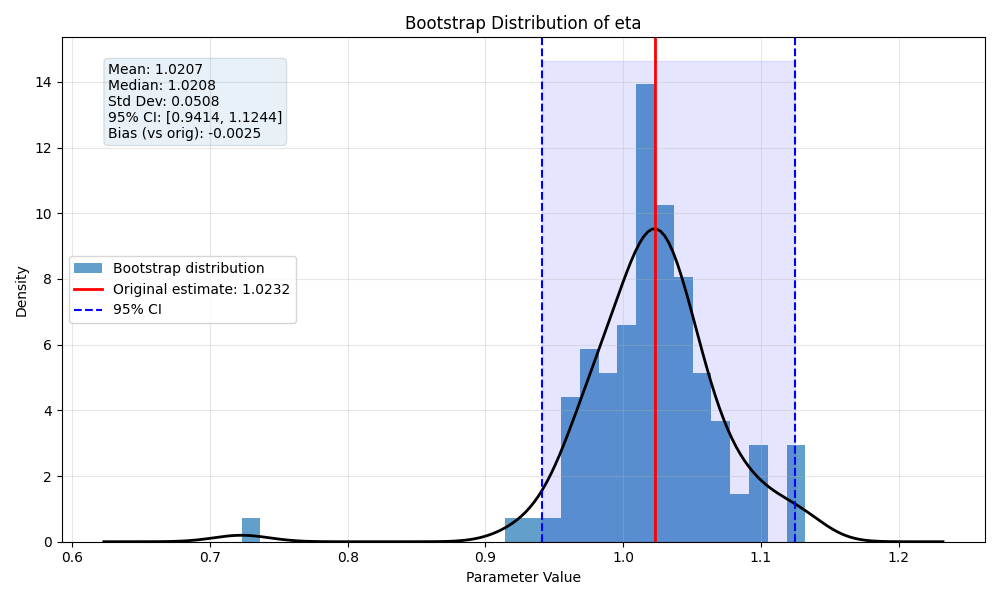}
         \caption{$\hat{\eta}$}
         \label{fig:eta_sparse}
     \end{subfigure}
        \caption{Estimated Parameters in Sparse Random Networks}
        \label{fig:est_sparse}
\end{figure}

\begin{figure}
     \centering
     \begin{subfigure}[b]{0.5\textwidth}
         \centering
         \includegraphics[width=\textwidth]{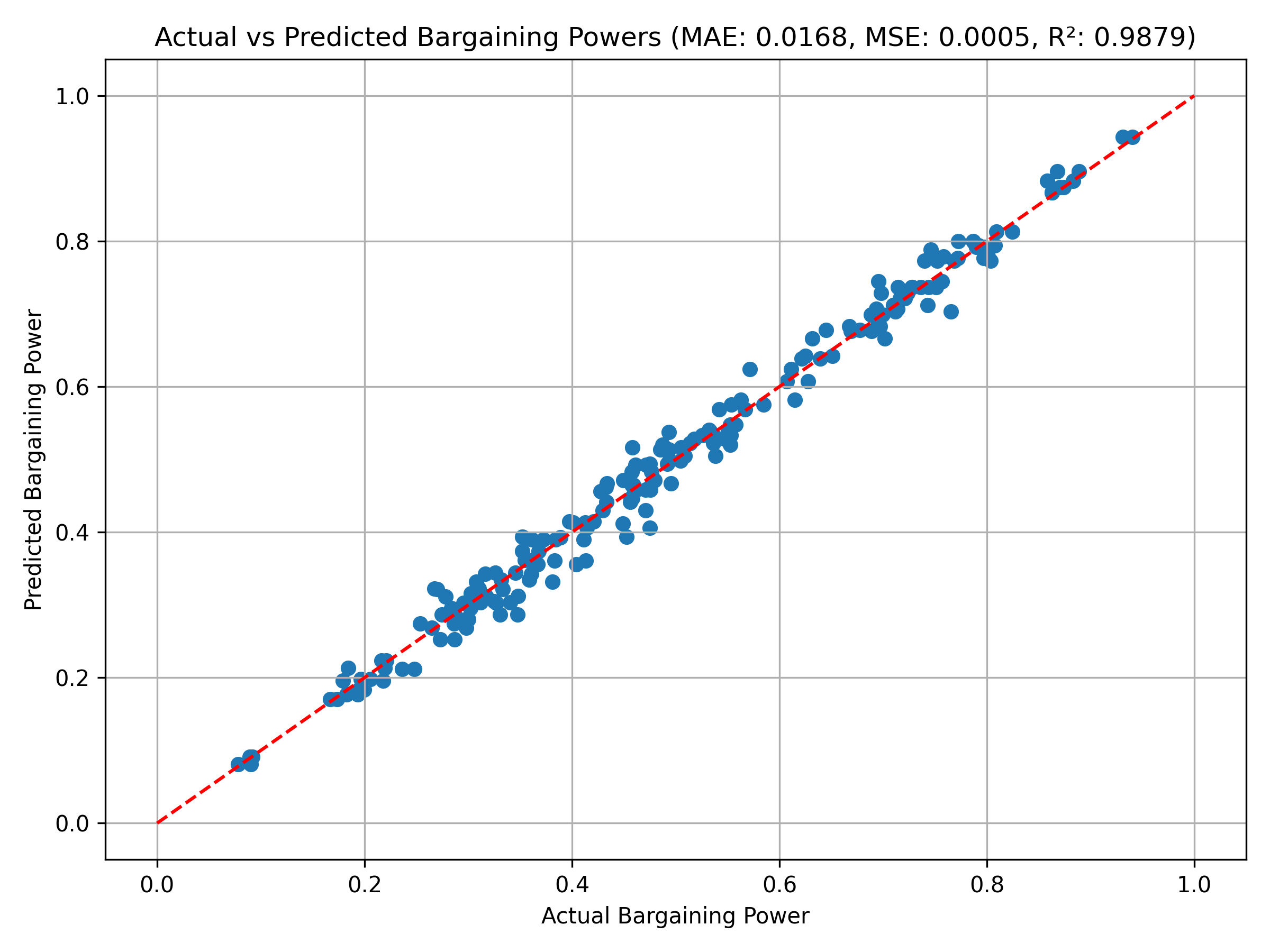}
         \caption{Bargaining Power}
         \label{fig:pi_sparse}
     \end{subfigure}%
     \begin{subfigure}[b]{0.5\textwidth}
         \centering
         \includegraphics[width=\textwidth]{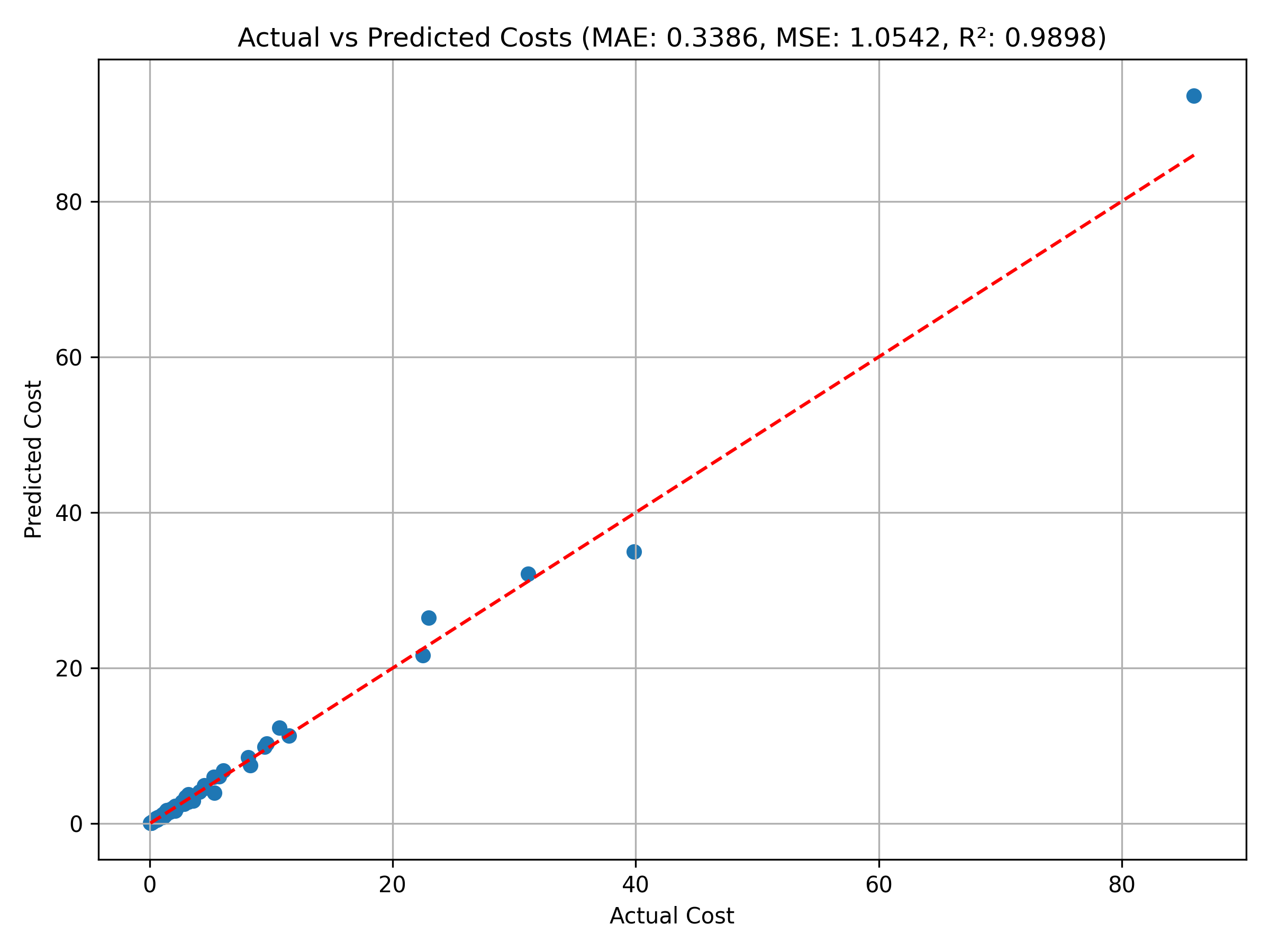}
         \caption{Holding Cost}
         \label{fig:c_sparse}
     \end{subfigure}
     \begin{subfigure}[b]{0.5\textwidth}
         \centering
         \includegraphics[width=\textwidth]{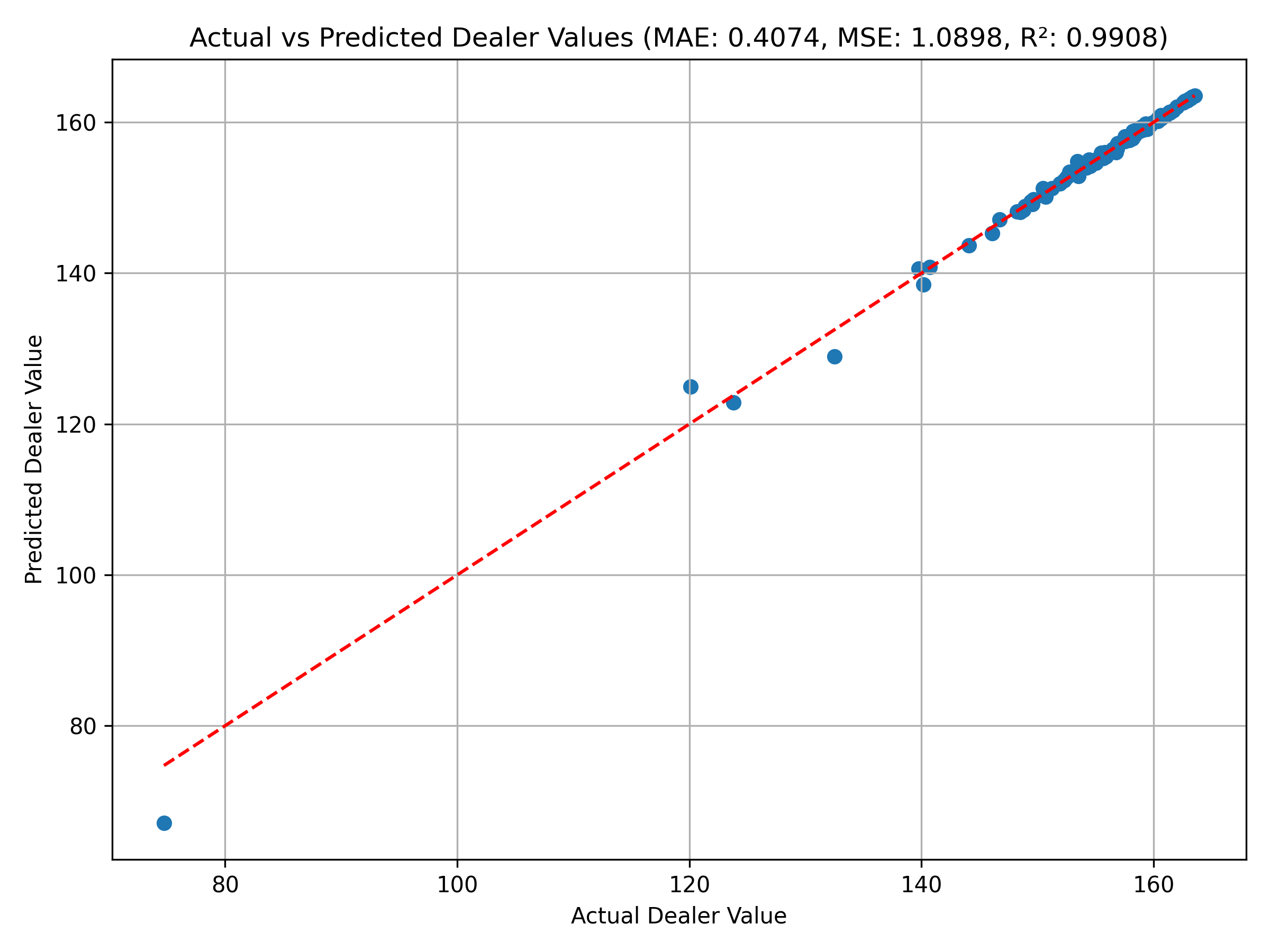}
         \caption{Dealer Values}
         \label{fig:v_sparse}
     \end{subfigure}%
          \begin{subfigure}[b]{0.5\textwidth}
         \centering
         \includegraphics[width=\textwidth]{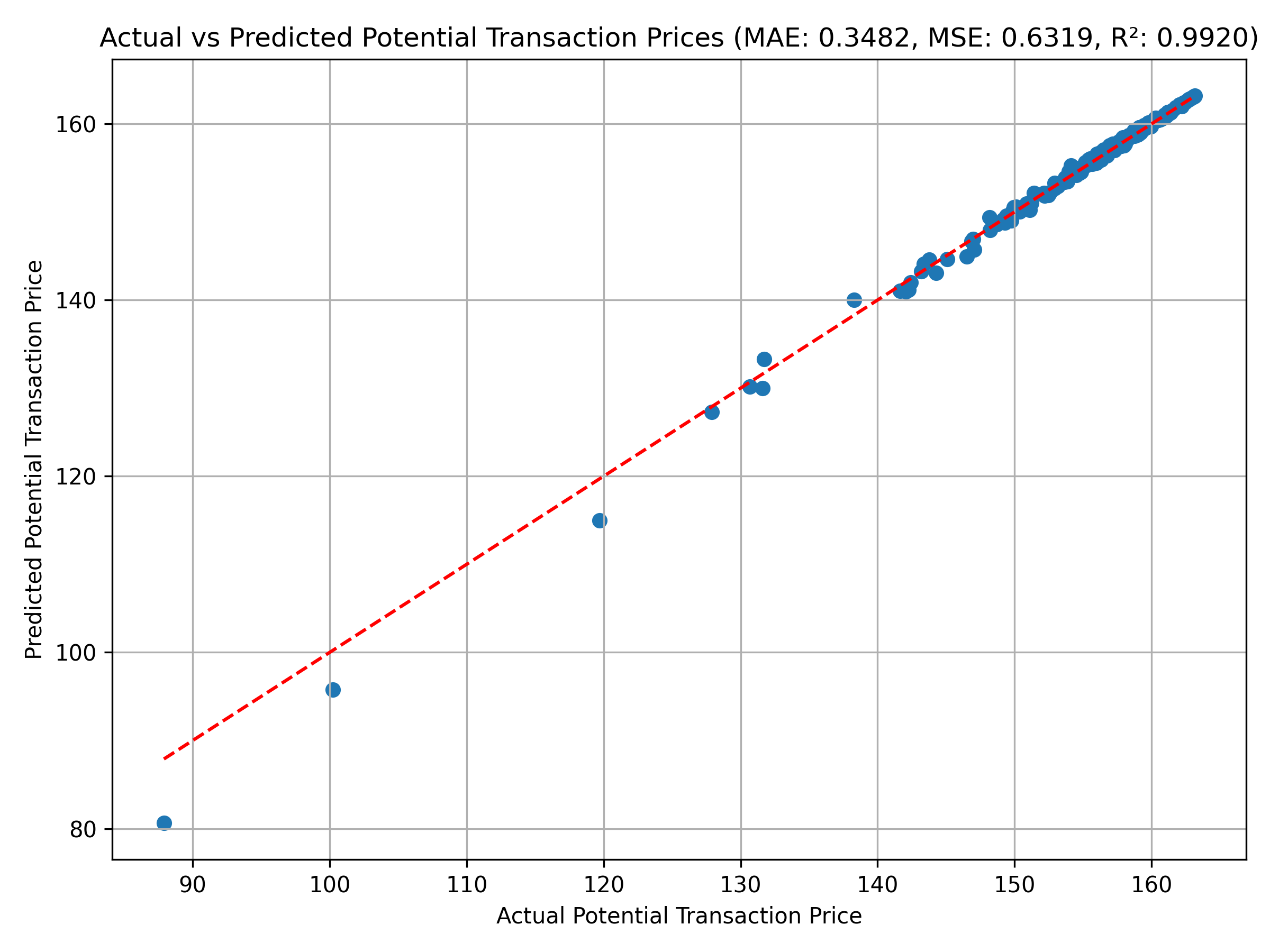}
         \caption{Potential Transaction Prices}
         \label{fig:p_sparse}
     \end{subfigure}
        \caption{Comparison of Predicted vs. Actual Latent Variables in Sparse Random Networks}
        \label{fig:est_latent_sparse}
\end{figure}

\begin{figure}
     \centering
     \begin{subfigure}[b]{0.6\textwidth}
         \centering
         \includegraphics[width=\textwidth]{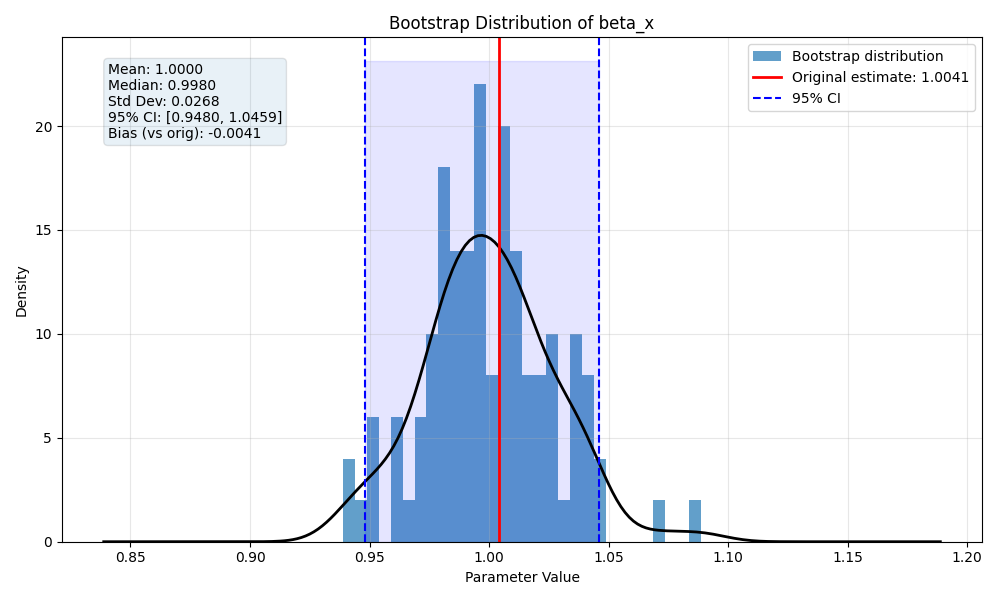}
         \caption{$\hat{\beta}_x$}
         \label{fig:beta_x_cp}
     \end{subfigure}
     \hfill
     \begin{subfigure}[b]{0.6\textwidth}
         \centering
         \includegraphics[width=\textwidth]{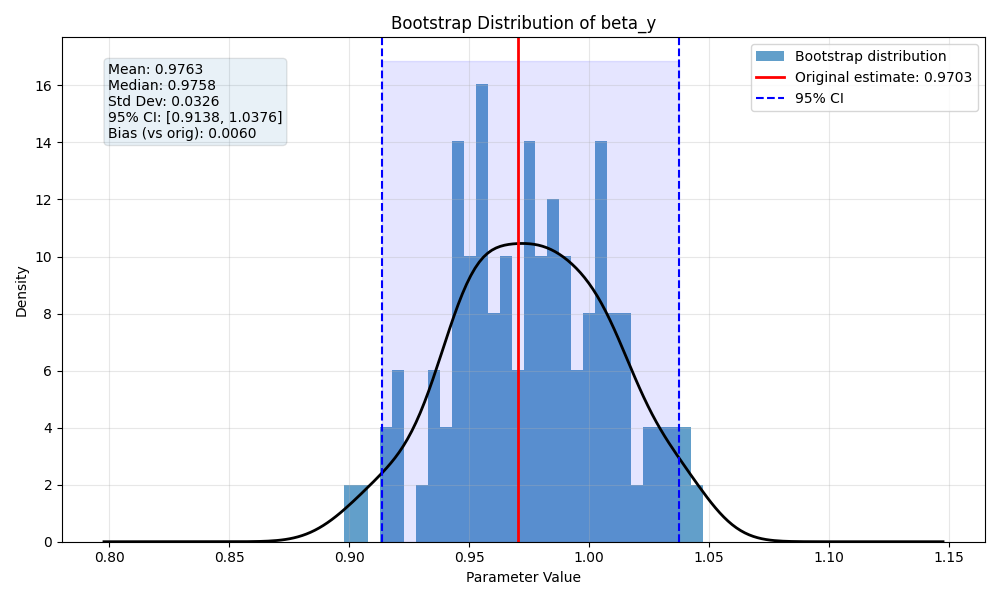}
         \caption{$\hat{\beta}_y$}
         \label{fig:beta_y_cp}
     \end{subfigure}
     \hfill
     \begin{subfigure}[b]{0.6\textwidth}
         \centering
         \includegraphics[width=\textwidth]{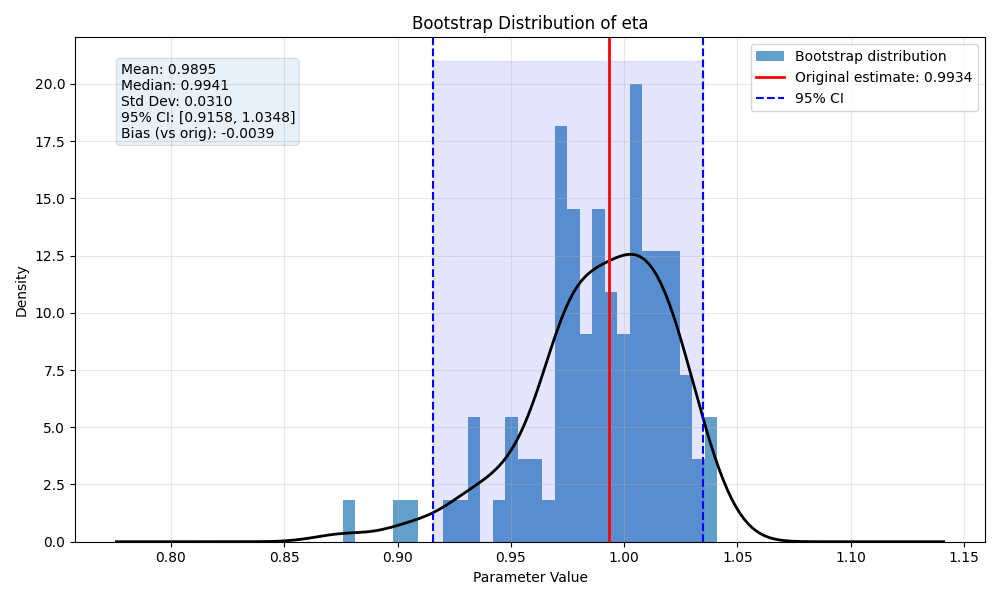}
         \caption{$\hat{\eta}$}
         \label{fig:eta_cp}
     \end{subfigure}
        \caption{Estimated Parameters in Core-periphery Networks}
        \label{fig:est_cp}
\end{figure}

\begin{figure}
     \centering
     \begin{subfigure}[b]{0.5\textwidth}
         \centering
         \includegraphics[width=\textwidth]{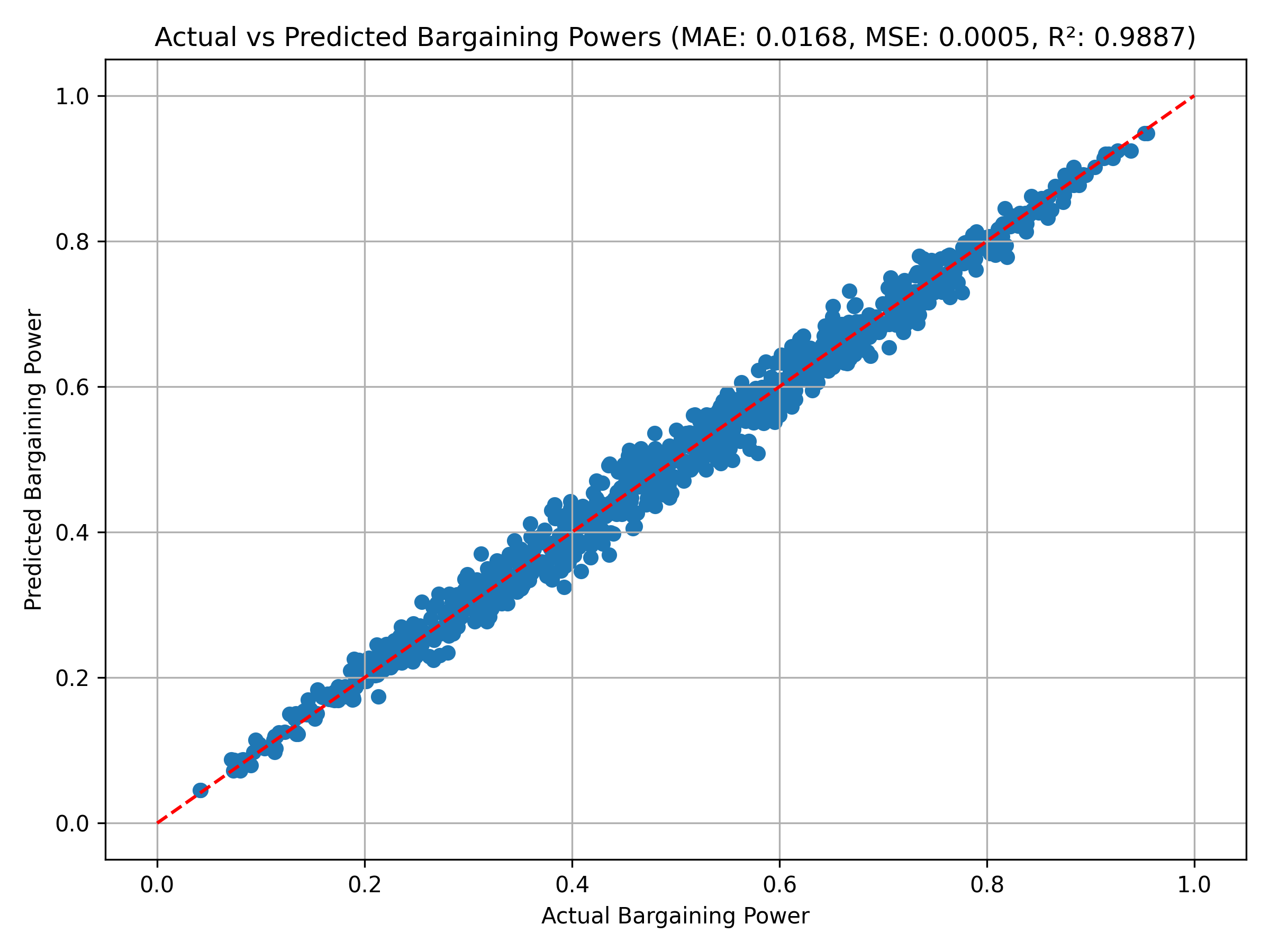}
         \caption{Bargaining Power}
         \label{fig:pi_cp}
     \end{subfigure}%
     \begin{subfigure}[b]{0.5\textwidth}
         \centering
         \includegraphics[width=\textwidth]{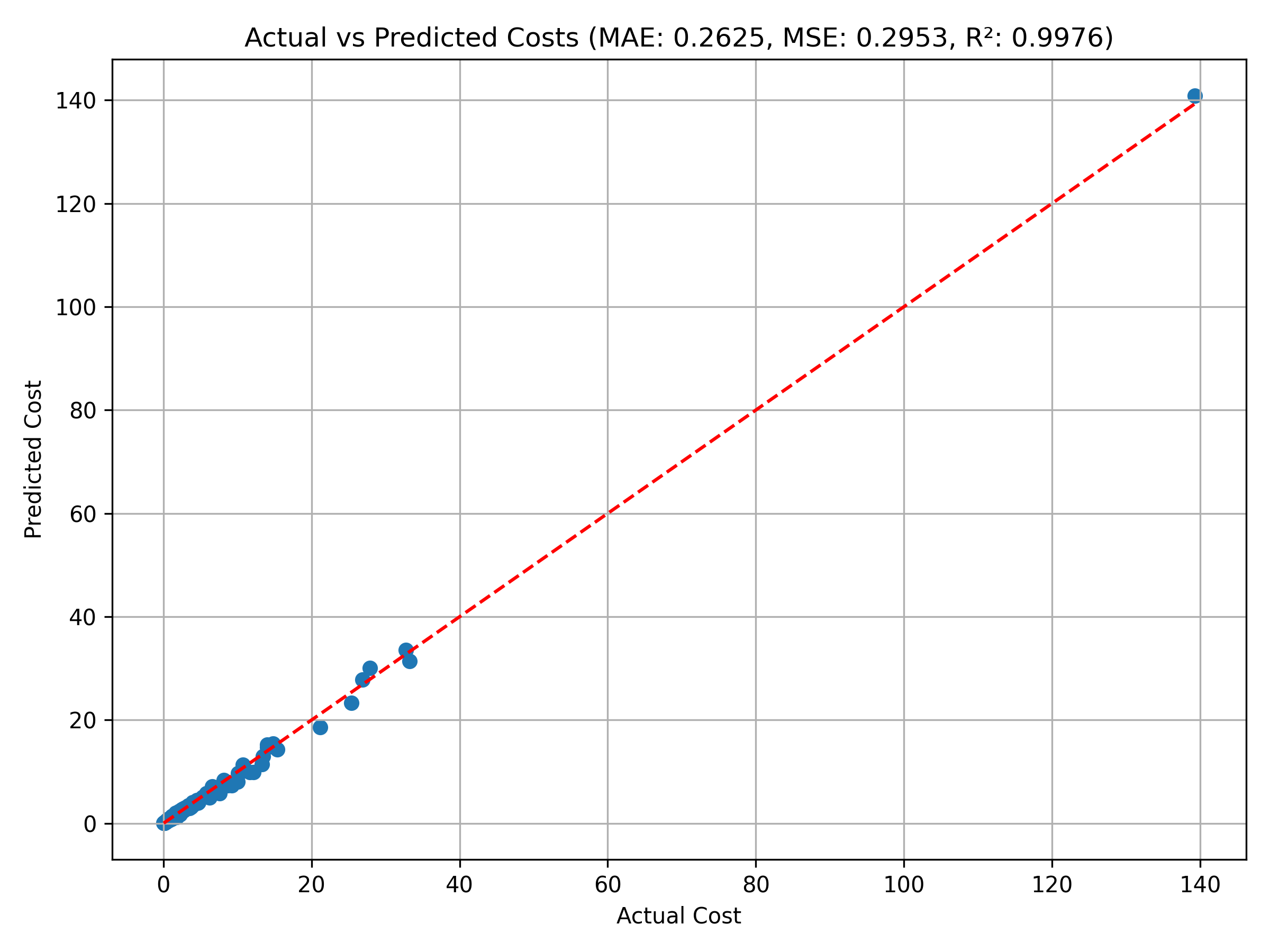}
         \caption{Holding Cost}
         \label{fig:c_cp}
     \end{subfigure}
     \begin{subfigure}[b]{0.5\textwidth}
         \centering
         \includegraphics[width=\textwidth]{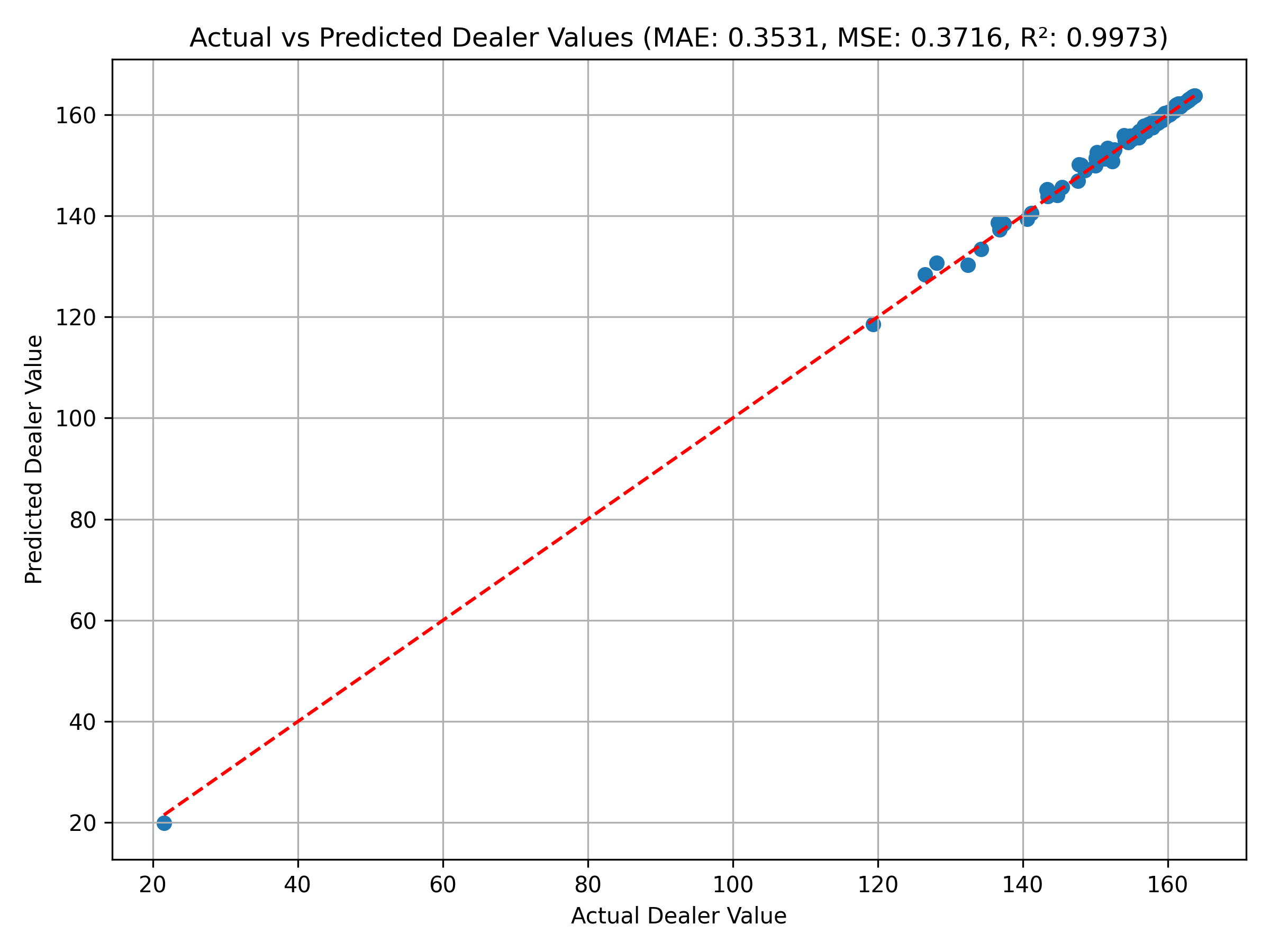}
         \caption{Dealer Values}
         \label{fig:v_cp}
     \end{subfigure}%
          \begin{subfigure}[b]{0.5\textwidth}
         \centering
         \includegraphics[width=\textwidth]{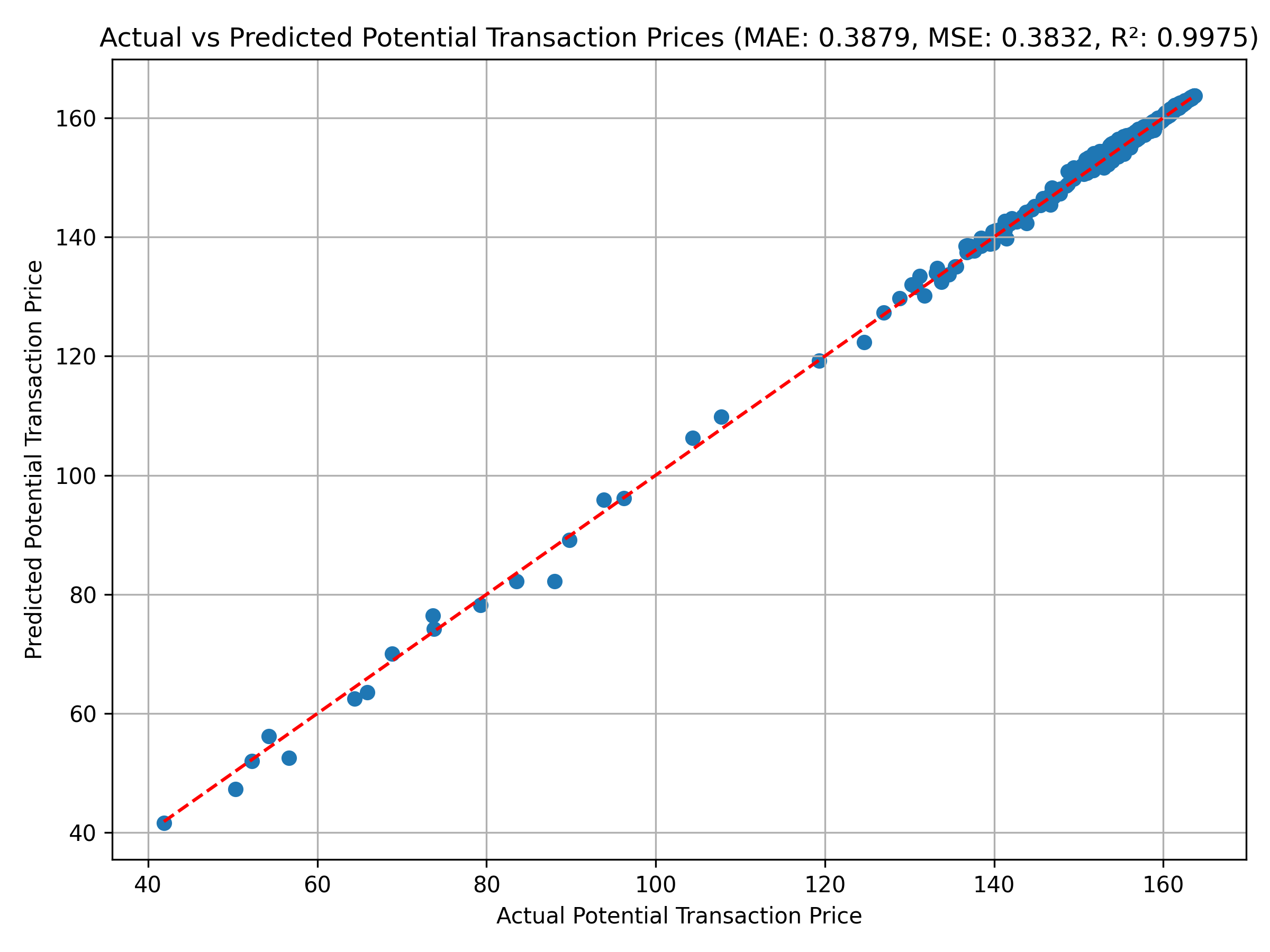}
         \caption{Potential Transaction Prices}
         \label{fig:p_cp}
     \end{subfigure}
        \caption{Comparison of Predicted vs. Actual Latent Variables in Core-periphery Networks}
        \label{fig:est_latent_cp}
\end{figure}


\end{document}